\documentclass[lettersize,journal]{IEEEtran}
\usepackage{amsmath,amsfonts}
\usepackage{array}
\usepackage[caption=false,font=footnotesize,labelfont=bf,textfont=footnotesize,labelsep=space]{subfig}
\usepackage{textcomp}
\usepackage{stfloats}
\usepackage{url}
\usepackage{verbatim}
\usepackage{graphicx}
\def\BibTeX{{\rm B\kern-.05em{\sc i\kern-.025em b}\kern-.08em
    T\kern-.1667em\lower.7ex\hbox{E}\kern-.125emX}}
\usepackage{balance}

\usepackage{algorithm}
\usepackage{algpseudocode}
\usepackage{cite}
\usepackage{amsthm}

\usepackage{caption}
\usepackage{subcaption}
\usepackage{booktabs}
\newtheorem{theorem}{Theorem}

\IEEEoverridecommandlockouts

\ifCLASSINFOpdf

\else

\fi

\usepackage{hyperref}

\hypersetup{
    colorlinks=true, 
    linkcolor=blue, 
    citecolor=red, 
    urlcolor=cyan, 
    filecolor=magenta, 
    bookmarks=true, 
    bookmarksopen=true, 
}

\begin{document}

\title{Task-Agnostic Learnable Weighted-Knowledge Base Scheme for Robust Semantic Communications}
\author{Shiyao~Jiang,
        Jian~Jiao,~\IEEEmembership{Senior Member,~IEEE,}
        Xingjian~Zhang,~\IEEEmembership{Member,~IEEE,}
        Ye~Wang,~\IEEEmembership{Member,~IEEE,}
        Dusit~Niyato,~\IEEEmembership{Fellow,~IEEE,}
        and~Qinyu~Zhang,~\IEEEmembership{Senior Member,~IEEE}

\thanks{Shiyao Jiang, Jian Jiao, Xingjian Zhang, and Qinyu Zhang are with the Guangdong Provincial Key Laboratory of Aerospace Communication and Networking Technology, Harbin Institute of Technology (Shenzhen), Shenzhen 518055, China, and also with the Pengcheng Laboratory, Shenzhen 518055, China (e-mail: jiang\_shiyao@foxmail.com; jiaojian@hit.edu.cn; x.zhang@hit.edu.cn; zqy@hit.edu.cn).}
\thanks{Ye Wang is with Pengcheng Laboratory, Shenzhen 518055, China~(e-mail: wangy02@pcl.ac.cn).}	
\thanks{Dusit Niyato is with the College of Computing and Data Science, Nanyang Technological University, Singapore (email: dniyato@ntu.edu.sg).} 

}

\maketitle

\begin{abstract}

With the emergence of diverse and massive data in the upcoming sixth-generation (6G) networks, 
the task-agnostic semantic communication system is regarded to provide robust intelligent services. 
In this paper, we propose a task-agnostic learnable weighted-knowledge base semantic communication (TALSC) framework for robust image transmission to address the real-world heterogeneous data bias in KB, including label flipping noise and class imbalance.
The TALSC framework incorporates a sample confidence module (SCM) as \emph{meta-learner} and the semantic coding networks as \emph{learners}. 
The \emph{learners} are updated based on the empirical knowledge provided by the learnable weighted-KB (LW-KB). 
Meanwhile, the \emph{meta-learner} evaluates the significance of samples according to the task loss feedback, and adjusts the update strategy of \emph{learners} to enhance the robustness in semantic recovery for unknown tasks.
To strike a balance between SCM parameters and precision of significance evaluation, 
we design an SCM-grid extension (SCM-GE) approach by embedding the Kolmogorov-Arnold networks (KAN) within SCM, 
which leverages the concept of spline refinement in KAN and enables scalable SCM with customizable granularity without retraining.
Simulations demonstrate that the TALSC framework effectively mitigates the effects of flipping noise and class imbalance in task-agnostic image semantic communication, achieving at least 12\% higher semantic recovery accuracy (SRA) and multi-scale structural similarity (MS-SSIM) compared to state-of-the-art methods.

\end{abstract}


\section{Introduction}

In the upcoming sixth-generation (6G) networks, semantic communication (SemCom) enables efficient and intelligent transmission by focusing on conveying the meaning of information rather than raw data, and supports various advanced services, such as object detection, image classification and segmentation \cite{6G, AIforSem}. 
Existing task-specific SemCom systems typically leverage Joint Source and Channel Coding (JSCC) with a well-aligned Knowledge Base (KB) to capture task-specific knowledge for tailored transmissions \cite{Task-SC, Task-SC2, Task-SC3}.
However, these paradigms that employ a model for a specific task might be limited, since the semantic coding network has to be updated once the task or environment is changed.
In light of these challenges, task-agnostic semantic communication (TASC) has the potential to support diverse intelligent services by incorporating limited receiver feedback without the need for task-specific refinements.

The ability to achieve robust image transmission for TASC relies on maintaining alignment between the KB and actual distribution of the observation space \cite{SC-survey, SC-survey2, GSC}. 
However, the emergence of heterogeneous data biases in KB complicates this alignment, particularly when environments and tasks are diverse. 
These biases may arise from variations in data sources \cite{ImgNet, imbalancedLearning,noisyLabel}, differences in collection conditions \cite{dataOpt,dataCollect}, or task-induced shifts in data patterns \cite{transferLearning, transferLearning2}. 
As a result, semantic coding networks often overfit to biased KBs, leading to poor performance during transmission.
On one hand, adaptive and weighted learning methods address this challenge by dynamically adjusting the importance of samples based on their relevance or reliability \cite{L2RW, metaLearnedConfidence}. 
While effective in mitigating data bias, these methods require prior knowledge or predefined rules to assign appropriate weights.
In contrast, robust learning methods for SemCom mitigate the impact of biased KBs by incorporating contaminated samples into the training dataset and considering the influence of semantic noise during the training process \cite{VQ-VAE, adv-robust}.
These methods help the model cope with biased KBs but tend to incur significant computational overhead, and may face scalability challenges when applied to diverse, evolving data bias distributions.
Meta-learning, particularly through the use of hyper-parameter prediction functions conditioned on tasks \cite{learn-hyperpred}, can offer a more flexible solution by enabling the evaluation of sample significance in KB.
This allows the model to prioritize task-relevant or high-quality data, improving robustness against heterogeneous data biases. 
Therefore, meta-learning methods stand out for their ability to rapidly adapt to varying data bias distributions with minimal adjustments, reducing the need for frequent re-training or KB updates.

\subsection{Related Work}

Achieving efficient TASC requires a robust KB capable of addressing the bias distributions between the statistical characteristics of the data to be transmitted and those of the empirical data in KB.
While updating the KB by incremental learning \cite{incremental-learning} and domain adaptation \cite{domainAdaptation} can improve performance, it typically involves significant communication and computational overheads. 
Ideally, the receiver can feed all of its local empirical semantic information that is relevant to its personalized tasks back to the transmitter, but this procedure incurs high costs and raises privacy concerns. 
Although robust learning methods enhance system performance, they often require additional re-training resources  \cite{VQ-VAE, adv-robust, robustLearning}. 
To minimize retraining overhead, a receiver-leading semantic coding framework is employed in conjunction with transmitter-side domain adaptation \cite{task-unaware}, though performance may degrade under severe task-data misalignment.  
Motivated by recent advances in meta-learning \cite{metaLearningReview}, several works have been proposed to automatically learn the adaptive weighting schemes with the aid of additional meta data \cite{MentorNet, MWNet,  CMWNet}.
Although these schemes do not explicitly consider TASC scenarios, they reveal the potential of meta-learning to handle the heterogeneous data biases in KB.

Recently, the authors in \cite{learn-hyperpred} proposed a framework for learning hyper-parameter prediction functions, which can guide the evaluation of sample significance and facilitate more effective updates to the KB in dynamic environments.
Similarly, the authors in \cite{AI-task} proposed constructing the KB by assigning importance weights to the feature maps (i.e., semantic information) for each class. 
This approach allows the semantic coding network to prioritize the transmission of important feature maps by jointly considering both bandwidth and performance requirements for robust transmission.
However, they only considered gradients as a measure of sample significance in KB, which has limitations in capturing more complex relationships between samples. 
Therefore, a weighted scheme for sample significance evaluation and KB updates is essential to guide semantic coding in TASC, ensuring robustness against heterogeneous data biases. 

Moreover, to construct a learnable weight mapping, the authors in \cite{MWNet} proposed parameterizing the weighting function as a multilayer perceptron (MLP) network, which benefits from the flexibility and generalization ability granted by the universal approximation theorem \cite{UAT}. 
However, MLP network requires large amounts of empirical data to avoid overfitting and may struggle with scalability in high-dimensional tasks. 
In contrast, inspired by the evolution of Kolmogorov-Arnold Representation Theorem (KART) \cite{KART, KART-study, KART-study2, KART-study3}, 
Kolmogorov-Arnold Networks (KANs) provide stronger theoretical guarantees for modeling complex, high-dimensional relationships, but they face challenges in high computational complexity in training \cite{KAN, KAN-survey}. 
Thus, MLPs are well-suited for tasks where flexibility and quick adaptation are needed, whereas KANs offer advantages in environments requiring more robust, theoretically grounded function approximations, but at the cost of higher computational demands for training.
This highlights the need to select the most appropriate architecture based on the specific requirements of the task, balancing between generalization ability and computational efficiency.

\begin{figure*}[t]
	\centering
	\includegraphics[width=0.9\textwidth]{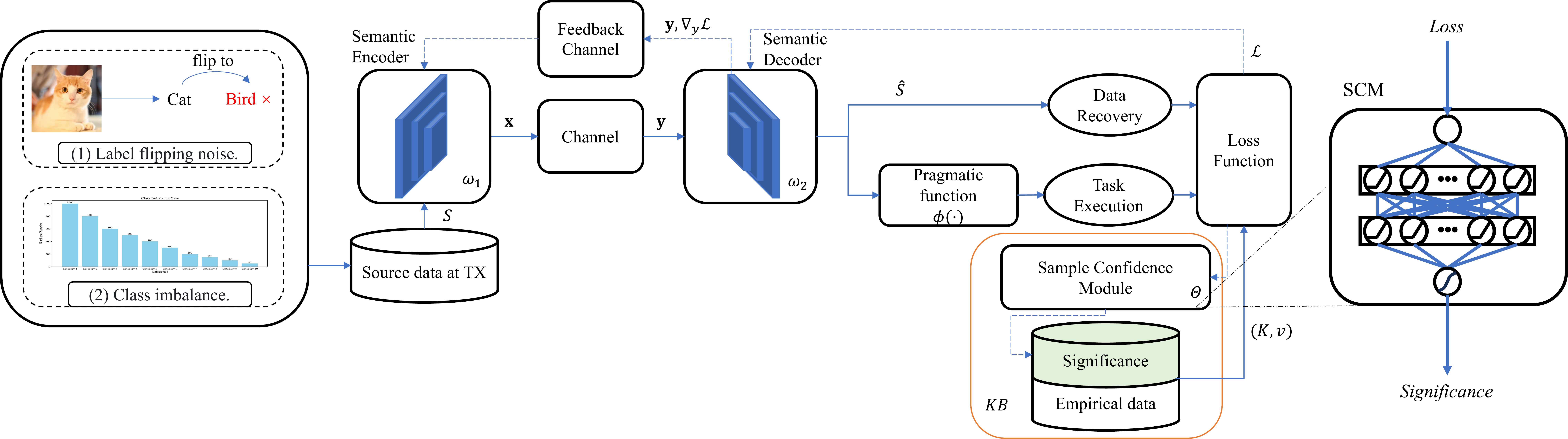}
	\caption{Our TALSC framework for robust image semantic transmission.
	At the transmitter side, when the receiver initiates a update request, the transmitter samples the empirical samples $ S $ from $\mathcal{K}$ and is encoded by $ f(\cdot) $ to obtain the channel input $ \mathbf{x} $,  and $ \mathbf{x} $ is transmitted to the receiver.  
	At the receiver side, the channel output is $ \mathbf{y} $,  and the receiver recovers the distorted encoded information $ \widehat{\mathbf{x}} $ based on CSI.
	Then, based on the decoder $ g(\cdot) $, the distorted empirical samples $ \widehat{S} $ is reconstructed as $\widehat{S} = g(\widehat{\mathbf{x}})$.
The distorted output of the pragmatic task is reconstructed through the pragmatic function $\phi(\cdot)$ as $\widehat{Z} = \phi(\widehat{S})$. 
Subsequently, the SCM estimates the significance of $ S $ to compute a weighted loss $\mathcal{L}$. 
Finally, receiver-driven updates for the semantic coding networks are conducted through the feedback channel.
	}
	\label{sys}
	\vspace{-8pt}
\end{figure*}

\subsection{Motivations and Contributions}
Motivated by above, we propose a Task-Agnostic Learnable weighted-knowledge base for Semantic Communication (TALSC) framework for robust image transmission.
Our contributions are summarized as follows.

\begin{itemize}

\item 
We propose a TALSC framework for robust image semantic transmission, which incorporates a well-designed sample confidence module (SCM) as \emph{meta-learner}, and the semantic decoding network as \emph{learner}. 
In each transmission, the receiver employs a pragmatic function to calculate the \emph{task loss} based on the output of \emph{learner}. 
We implement a significance evaluation function (SEF) in \emph{meta-learner} to map task loss to determine the significance of each sample in KB.
Then, the \emph{learner} is temporarily updated based on the sample significance and \emph{task loss}, which can suppress low-significance samples due to label flipping noise and class imbalance data.
Subsequently, the receiver applies the pragmatic function to calculate the \emph{meta-loss} by the updated output from \emph{learner}, and the SCM is updated according to \emph{meta-loss}. 
The updated SCM recalculates the significance by utilizing the \emph{task loss}, and the \emph{learner} is finally updated based on the recalculated sample significance and \emph{task loss}. 
The updated \emph{learner} outputs the final result in this transmission.
The proposed TALSC framework enables adaptive and robust semantic communications by integrating metadata-guided sample significance evaluation with knowledge-aware learner updates, effectively enhancing generalization across heterogeneous and noisy data distributions.

\item 
We first construct the SEF in SCM via the MLPs for high practicality, 
and analyze the approximation capability of MLP-based SEF for general functions by the Universal Approximation Theorem (UAT).
To represent the SEF more efficiently and accelerate its optimization process with greater flexibility, 
we construct the SEF by utilizing the KAN networks. 
By leveraging the spline-based structure of KAN, we derive an upper bound on the approximation error of KAN-based SEF, 
and the decay rate of approximation error is determined by the smoothness of the target SEF and the spline order. 
We further propose a Grid Extension (GE) approach for KAN-based SEF to achieve higher resolution and capture fine-grained details for lower \emph{meta-loss}. 
The GE approach first utilizes a low-dimensional parameter set to construct a coarse grid SEF. 
Then, we derive an analytical transformation matrix to extend a finer grid SEF according to the minimal mean square error (MMSE) criterion.
Finally, fine-tuning parameters of the extended fine grid SEF further refines the representation of \emph{meta-learner}, thereby reducing the average \emph{meta-loss}.

\item
We conduct extensive simulations to demonstrate the efficiency and robustness of our TALSC framework.
First, we show that the KAN-based SCM achieves lower average \emph{meta-loss} compared to the MLP-based SCM under equivalent parameter sizes, indicating that KAN-based representations offer superior efficiency in adjusting the \emph{learner}.
Additionally, 
the GE approach extended fine grid SEF can obtain lower \emph{meta-loss} than that of coarse grid SEF, while maintaining similar costs  in terms of training time and memory usage.
Furthermore, we conduct simulations for our TALSC framework versus the label flipping noise, i.e., impact of varying flipping noise rates (FNRs) on \emph{learner}.
The results demonstrate that the TALSC framework can maintain high semantic recovery accuracy (SRA) and multi-scale structural similarity index measure (MS-SSIM) even under high FNR compared to task agnostic DeepJSCC. 
Finally, simulation results for class imbalance also reveal that the TALSC framework enhances the F1-score for minority classes, thereby improving the extraction of semantic information from underrepresented classes.

\end{itemize} 

The rest of this paper is structured as follows.
Section \ref{LW-KB} introduces the TALSC framework for robust image semantic transmission.
In Section \ref{SCM-KAN}, we construct the SCM based on MLP and KAN, respectively, and design the Grid Extension approach for the KAN-based SCM.
In Section \ref{sim}, we simulate different cases of KB biases and prove the effectiveness of our TALSC framework compare to the state-of-the-art schemes.
Finally, we draw our conclusions in Section \ref{conclusion}.

\section{System Model and TALSC Framework} \label{LW-KB}

In this section, we present our TALSC framework, which comprises the semantic coding networks for transmission and an SCM to evaluate the significance of KB source data.  
The KB that integrates a learnable significance for its empirical data is referred to as the learnable weighted KB (LW-KB).
The transmitter and receiver are regarded as unequal participants in our TALSC framework, where the receiver is the dominant one with more resources and data authorization privilege, and we assume that the receiver has enough power to achieve a noiseless feedback channel.
In Section \ref{A}, we introduce the transmission model for semantic coding networks, including the task-agnostic pragmatic function, the semantic coding network based on convolutional neural networks (CNNs), and the modeling of three typical communication channels. 
They include additive white Gaussian noise (AWGN), Rayleigh, and Rician channels. 
The adaptive weighting design for semantic coding network updates is discussed in Section \ref{B}, where we introduce the weighted loss functions to account for the significance during updates.
Furthermore, considering task unawareness and privacy preservation at the transmitter \cite{task-unaware}, we establish a receiver-driven process to guide transmitter updates.
Section \ref{C} introduces the evaluation model and optimization of the SCM, while Section \ref{D} proposes the update mechanism of TALSC framework.

\subsection{Transmission Model for Semantic Coding Networks} \label{A}

The system flow for robust task-agnostic image semantic transmission is shown in Fig. \ref{sys}.
Let $ f(\cdot) $ and $ g(\cdot) $ represent the semantic encoder and decoder, respectively. 
We model the semantic encoder $ f_{\omega_1}(\cdot) $ and decoder $ g_{\omega_2}(\cdot) $ as two deep neural networks (DNNs), where $\omega_1$ and $\omega_2 $ are their corresponding parameters, respectively.
There are no specific requirement on the network structure, and we use the CNNs \cite{CNN} to implement the semantic coding networks.
The pragmatic function at the receiver serves as a flexible mathematical abstraction that encapsulates the core principles of tasks to enable TASC,  
and we focus on classification and image recovery tasks in this paper. 
We employ a pre-trained lightweight GoogleNet as the pragmatic function \cite{GoogleNET}, which has efficient inception modules and depthwise separable convolutions to provide robust support for classification tasks, while maintaining the adaptability.
For the communication channels, we primarily focus on the AWGN channel, and also analyze the Rayleigh and Rician fading channels. 
The complex symbol $ \mathbf{x} \in \mathbb{C}^{B \times 1} $ is encoded by the semantic  encoding network at the transmitter and transmitted over the AWGN channel, 
and the received symbol $ \mathbf{y} \in \mathbb{C}^{B \times 1} $ can be expressed as:
\begin{equation}
\mathbf{y} = \mathbf{x} + \mathbf{n}, 
\end{equation}
where $ \mathbf{n} $ represents $ \mathcal{CN}(0, \sigma_n^2) $ AWGN with $ \sigma_n^2 $ as the  variance. 

For the  block fading channel, the received signal $ \mathbf{y} $ can be expressed as:
\begin{equation}
\mathbf{y} = h \mathbf{x} + \mathbf{n}, 
\end{equation}
where $ h $ denotes the channel coefficient and remains constant during the transmission of a batch of images, 
and $ h $ follows a complex normal distribution $ \mathcal{CN}(0, 1) $ as Rayleigh fading  \cite{VQSC, DeepJSCC-f}.
In Rician fading channel,  $ h $ is a complex normal distribution $ \mathcal{CN}(\mu, \sigma^2) $ with $ \mu = \sqrt{\frac{r}{r+1}} $ and $ \sigma = \sqrt{\frac{1}{r+1}} $, where $r$ is the Rician coefficient \cite{VQSC}. 
Assuming that the receiver has perfect channel state information (CSI), the recovered symbol can be expressed as:
\begin{equation} \label{x-recover}
\widehat{\mathbf{x}} = \frac{\overline{h}\cdot \mathbf{y}}{\| h \|_2^2}, 
\end{equation}
where $\overline{h}$ represents the complex conjugate.

As shown in Fig. \ref{sys}, the task-agnostic image semantic transmission in TALSC framework is as follows. 
At the transmitter side, when the receiver initiates a data request, the transmitter samples the observational data $ S $ from the observation space, and  $ S $ is similar to the empirical data $ K $ in KB and is encoded by $ f(\cdot) $ to obtain the channel input $ \mathbf{x} $,  and $ \mathbf{x} $ is transmitted to the receiver.  
For RGB-format image sources, both  $ S $ and  $ K $ are modeled as three-dimensional vector sequences with 8-bit symbols.
At the receiver side, the channel output is $ \mathbf{y} $,  and the receiver recovers the distorted encoded information $ \widehat{\mathbf{x}} $ based on CSI.
Note that the power of the channel noise $\mathbf{n}$ is fixed but unknown. 
Then, based on the decoder $ g(\cdot) $, the distorted observable data $ \widehat{S} $ is reconstructed as
$\widehat{S} = g(\widehat{\mathbf{x}}).$
The distorted output of the pragmatic task is reconstructed through the pragmatic function $\phi(\cdot)$ as $\widehat{Z} = \phi(\widehat{S})$. This process ensures that the receiver generates a pragmatic output $\widehat{Z}$ from the reconstructed observable data $\widehat{S}$, enabling task execution despite noise and distortion in the task-agnostic image semantic transmission.

\subsection{Adaptive Weighting Design for Semantic Coding Network Update} \label{B}

In this part, we assume the existence of a perfect SCM capable of evaluating sample significance to reflect the value of each sample, prioritizing data with higher significance for updating the semantic coding networks. 
Consequently, the KB $\mathcal{K}$ consists of empirical data $ K_i $  along with their corresponding significance $v_i$ ($1\leq i \leq N$), which  enhancing the robustness and effectiveness of the TALSC framework, where $ N $ denotes the total number of empirical samples in KB. 
Details on the optimization of SCM are provided in Section \ref{C}.

Consider a semantic image transmission system involving a biased KB $\mathcal{K} = \{ K_i, v_i \}$ ($1\leq i \leq N$) and $ K_i = \{s_i, z_i\} $, where $ s_i $ represents the $ i $-th empirical sample in $\mathcal{K}$, and $ z_i \in \{0,1\}^C $ is the corresponding one-hot encoding. 
Note that $ s_i $ undergoes semantic encoding network $ f(\cdot) $, channel transmission, and semantic decoding network $ g(\cdot) $ in sequence, the task loss obtained from $ \phi(\cdot) $ can be expressed as:
\begin{equation} \label{pragmatic-loss}
L^{(\phi)}(s_i, z_i; \omega) = (1-\lambda)l_{CE}(\phi(\hat{s_i}), z_i)+\gamma \lambda l_{MSE}(\hat{s_i}, s_i),
\end{equation}
where $ l_{CE}(\cdot) $ and $ l_{MSE}(\cdot) $ denote the cross entropy (CE) loss and mean square error (MSE) loss, respectively, and $ \hat{s_i} = \hat{s_i}(s_i; \omega, h, \sigma_n^2) $ is the reconstructed observable data at the receiver.   
$ \gamma $ is a hyper-parameter to scale the value of MSE loss for the alignment with CE loss, and $ \lambda $ is a tradeoff ratio. Following the setup from \cite{task-unaware}, we define the tradeoff ratio as $\lambda=1-\text{compression rate}$.
For simplicity, $ L^{(\phi)}(s_i, z_i; \omega) $ is abbreviated as $ L_i^{(\phi)}(\omega) $ in the following.

Typically, the optimal semantic coding parameters $ \omega^* $ can be obtained by minimizing the expected $ L_i^{(\phi)}(\omega) $ as follows:
\begin{equation}
\omega^* = \arg\min_{\omega} \frac{1}{N} \sum_{i=1}^{N} L_i^{(\phi)}(\omega).
\end{equation}

In the presence of $\mathcal{K}$, the TALSC framework aims to enhance the robustness of semantic coding network by applying $ v_i \in [0,1] $ to $ L_i^{(\phi)}(\omega) $, and  $ \omega^* $ are then calculated by minimizing the following weighted loss function:
\begin{equation}
\omega^* = \arg\min_{\omega} \frac{1}{N} \sum_{i=1}^{N} v_i \cdot L_i^{(\phi)}(\omega).
\end{equation}

To make $ v_i$ more adaptable to the tasks and further improve the semantic coding performance and robustness, our TALSC primarily utilizes an SCM to learn an SEF from $ L_i^{(\phi)}(\omega) $ to $ v_i$, and dynamically applies $ v_i$ to adjust semantic coding networks in each transmission.    

Leveraging the concept of meta-learning, we take the SCM as \emph{meta-learner} with parameters $ \Theta $, while the semantic coding network serves as \emph{learner} with parameters $ \omega =\{\omega_1, \omega_2\} $. 
The \emph{learner} is trained to extract semantic information relevant to the receiver's pragmatic task, while $ \Theta $ adjusts and optimizes $\omega$ and configurations based on feedback from $\phi(\hat{s_i})$. 
The framework intelligently senses and quantifies the degree and type of heterogeneous data bias in $\mathcal{K}$ based on the feedback $\phi(\hat{s_i})$, subsequently amplifying the most relevant $K_i$ to the task and suppressing irrelevant noisy samples.

Furthermore, considering privacy and task-unawareness at the transmitter \cite{task-unaware}, a feedback channel is established to transmit some necessary values in order to obtain the update gradients for $ f_{\omega_1}(\cdot) $, as shown in Fig. \ref{sys}.
In the $t$-th transmission, the weighted loss $\mathcal{L}_t$ is defined by:
\begin{equation}
\mathcal{L}_t = \sum_{i \in M_t}v_i\cdot L_i^{(\phi)}(\omega),
\end{equation}
where $ M_t $ is the index set of the sampled batch of $K_i$ from $\mathcal{K}$.  
Based on $\mathcal{L}_t$ and  $\mathbf{y}$, the data tuple $ (\nabla_{\mathbf{y}}\mathcal{L}_t, \ \mathbf{y}) $ is constructed and feedback to the transmitter. 
Then, the transmitter calculates the update gradients for semantic encoder  $ f_{\omega_1}(\cdot) $ as follows:
\begin{equation} \label{feedback}
\nabla_{\omega_1}\mathcal{L}_t = \nabla_{\mathbf{y}}\mathcal{L}_t \cdot \nabla_{\omega_1}\mathbf{y},
\end{equation}
where $ \nabla_{\omega_1}\mathbf{y} $ can be calculated locally. The $\omega_1$ is then update as follows:
\begin{equation} \label{feedback-update}
\omega_1^{(t+1)} = \omega_1^{(t)} - \alpha \nabla_{\omega_1}\mathcal{L}_t.
\end{equation}
Through this receiver-driven update process, the transmitter remains unaware of the pragmatic use of the transmitted data, ensuring privacy and task-unawareness in task-agnostic image semantic communication system.

\begin{figure*}[t]
	\centering
	\includegraphics[width=0.8\textwidth]{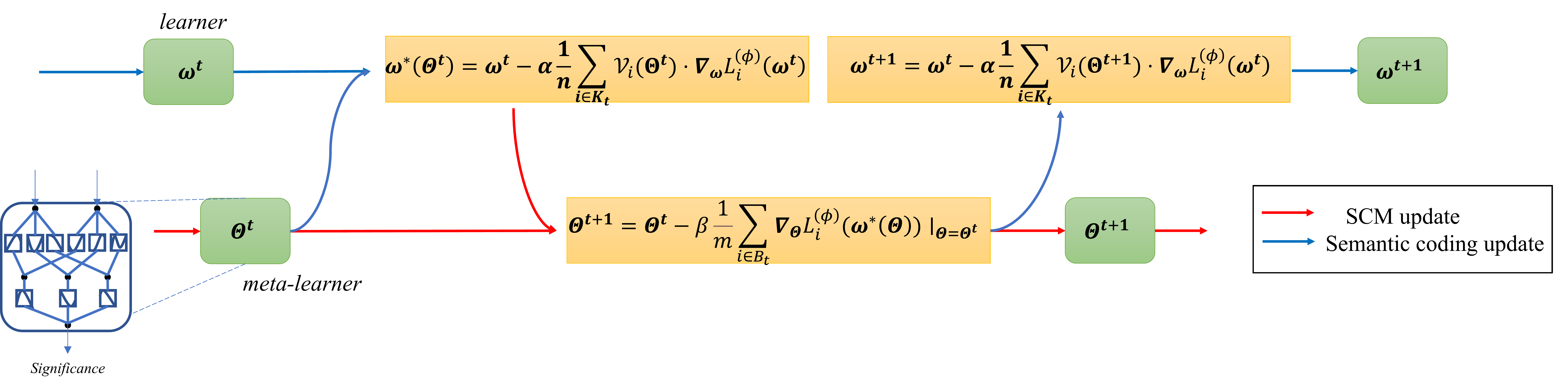}
	\caption{
Illustration of the flowchart for the update of TALSC framework from step $t$ to $t+1$. 
The meta-learner first performs a trial update on the learner to obtain temporary parameters $\omega^{*}(\Theta^{t})$. Based on this intermediate learner, the meta-learner is evaluated and updated according to equation (\ref{SCM-update}), yielding new parameters $\Theta^{t+1}$. Finally, the updated meta-learner is used to refine the learner, resulting in $\omega^{t+1}$.
	}
	\label{SCM-flowchart}
	\vspace{-8pt}
\end{figure*}

\subsection{Evaluation Model and Optimization of SCM} \label{C}

The goal of SCM is to learn an SEF from $\mathcal{K}$ to determine $v_i$ based on feedback from $\phi(\hat{s_i})$, where $v_i$ reflects the contribution of $K_i$ in $\mathcal{K}$ to the tasks of receiver.
Once the SEF is provided by SCM, the semantic coding networks can be further adjusted for robust semantic transmission.
To make the learned SEF more effective in evaluating and selecting $K_i$ in $\mathcal{K}$ and better update $\omega$, we model the evaluation of SCM weighting mapping as a bi-level optimization problem.

Note that we need a small amount of unbiased metadata, denoted as $ \mathcal{B} = \{s_j^{\text{meta}}, z_j^{\text{meta}}\}_{j=1}^M $.
Our objective is to learn an SEF $ \mathcal{V}(\cdot | \Theta) $ by optimizing the SCM parameters $ \Theta $.    

First, based on the initial $ \mathcal{V}(\cdot | \Theta) $, the optimal semantic coding network parameters $ \omega^*(\Theta) $ on $\mathcal{K}$ can be obtained by solving the following optimization problem:
\begin{equation}
\omega^*(\Theta) = \arg\min_{\omega} \mathbb{E}_{{\{s_i, z_i\}} \sim p_{\text{em}}(s, z)} \left(L_i^{(\phi)}(\omega) \cdot \mathcal{V}(L_i^{(\phi)}(\omega)\mid \Theta)\right),
\end{equation}
where $L_i^{(\phi)}(\omega)$ represents the \emph{task loss} for $K_i$, $ p_{\text{em}}(s, z) $ denotes the distribution of $K_i$ in $\mathcal{K}$,  and $\{s_i, z_i\}$ is drawn from $ p_{\text{em}}(s, z) $ to calculate $ L_i^{(\phi)}(\omega) $.
For simplicity, $ \mathcal{V}(L_i^{(\phi)}(\omega)\mid\Theta) $ is abbreviated as $ \mathcal{V}_i(\Theta) $ afterwards.

Then, we define the meta-loss $ \mathcal{L}^{\text{meta}}(\Theta) $ in terms of $ \Theta $ by calculating on the metadata $\mathcal{B}$ as follows:
\begin{equation}
\mathcal{L}^{\text{meta}}(\Theta) = \mathbb{E}_{\{s_i, z_i\} \sim p_{\text{meta}}(s, z)} \left(L_i^{(\phi)}(\omega^*(\Theta))\right),
\end{equation}
where $ p_{\text{meta}}(s, z) $ denotes the distribution of $\mathcal{B}$ and $\{s_i, z_i\}$ is drawn from $p_{\text{meta}}(s, z)$ to calculate $L_i^{(\phi)}(\omega^*(\Theta))$.

To find the optimal parameter $ \Theta^* $ to minimize $ \mathcal{L}^{\text{meta}}(\Theta) $,  a bi-level optimization problem is formulated as follows:
\setcounter{equation}{\theequation+1}
\begin{equation} \label{OP}
\begin{aligned}
\Theta^* 
&= \arg\min_{\Theta} \mathcal{L}^{\text{meta}}(\Theta) \\
&= \arg\min_{\Theta} \mathbb{E}_{\{s_i, z_i\} \sim p_{\text{meta}}(s, z)} \left(L_i^{(\phi)}(\omega^*(\Theta))\right),
\end{aligned}
\tag{\theequation a}
\end{equation}
s.t. \\
\begin{equation} \label{OP-constrain}
\begin{aligned}
\omega^*(\Theta) &= \arg\min_{\omega} \mathbb{E}_{{\{s_i, z_i\}} \sim p_{\text{em}}(s, z)} \left(L_i^{(\phi)}(\omega) \cdot \mathcal{V}_i(\Theta) \right),
\end{aligned}
\tag{\theequation b}
\end{equation}
where the meta-loss $ \mathcal{L}^{\text{meta}}(\Theta) $ also serves as a metric for evaluating the effectiveness of $\mathcal{V}(\cdot | \Theta)$. 
The optimization problem in the constraint (\ref{OP-constrain}) is the low-level optimization process, while the target optimization problem (\ref{OP}) is the high-level optimization process.

\begin{figure*}[b]
	\centering
	\includegraphics[width=0.7\textwidth]{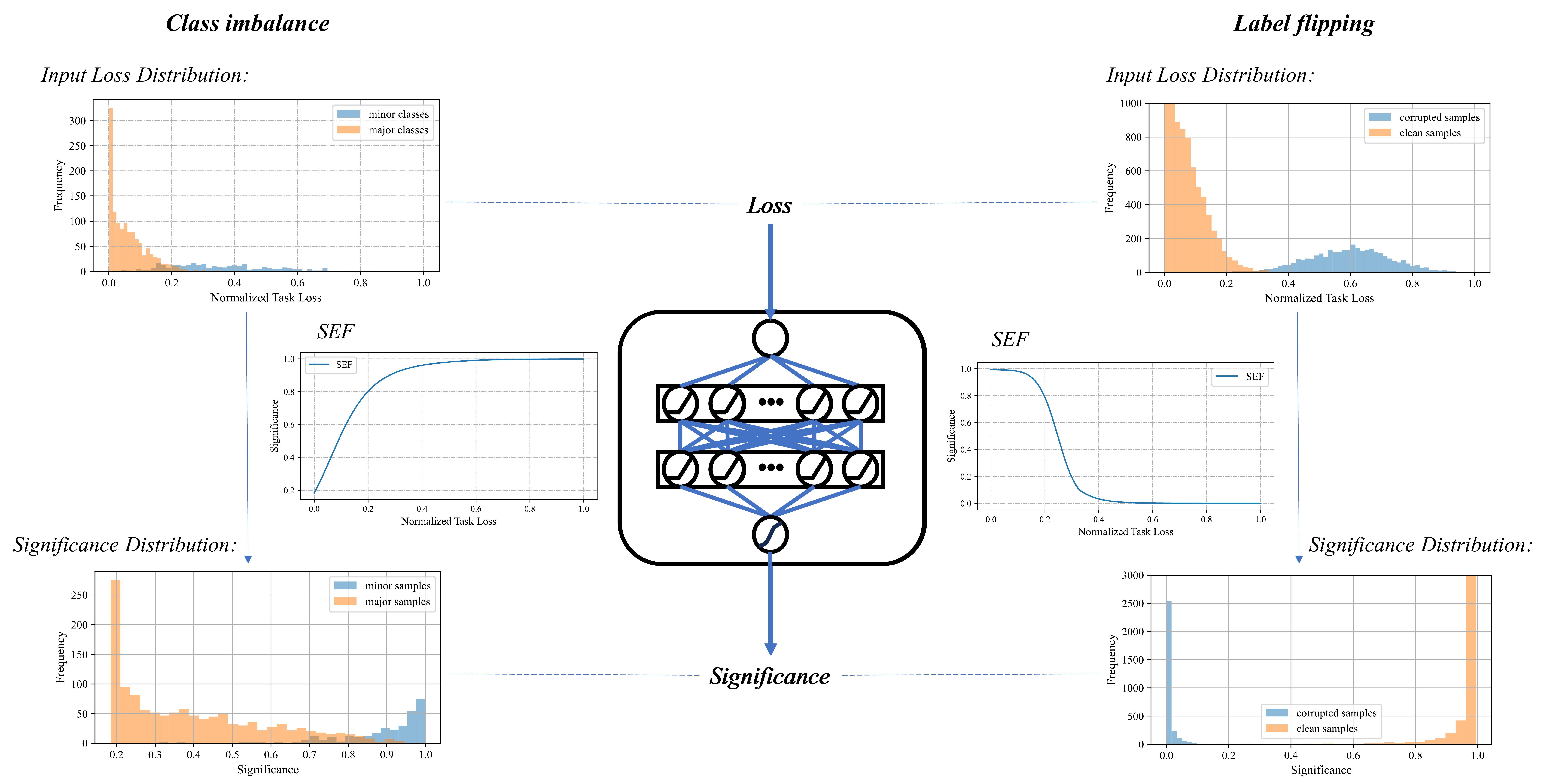}
	\caption{Meta essence understanding for our TALSC framework. 
For label flipping noise, the input loss distribution for $K_i$ demonstrates that corrupted samples typically have higher task loss $L_i^{(\phi)}(\omega)$ but contribute negligibly to the task. 
To address this, the SCM employs a monotonically decreasing SEF, which assigns higher significance values $v_i$ to clean samples with lower $L_i^{(\phi)}(\omega)$, while suppressing the influence of high-loss samples that are likely corrupted.
For class imbalance, the input loss distribution for $K_i$ reveals that the minority class samples often exhibit higher $L_i^{(\phi)}(\omega)$ due to their lack of representation in $\mathcal{K}$. 
To mitigate this issue and enhance the contribution of the minority class, the SCM adjusts SEF to assign greater $v_i$ to these underrepresented samples. 
This reweighting ensures that the minority class receives appropriate attention for semantic coding networks, thereby improving transmission robustness across classes.
}
	\label{meta-essence}
	\vspace{-8pt}
\end{figure*}

\begin{figure*}[t]
	\centering
	\includegraphics[width=0.55\textwidth]{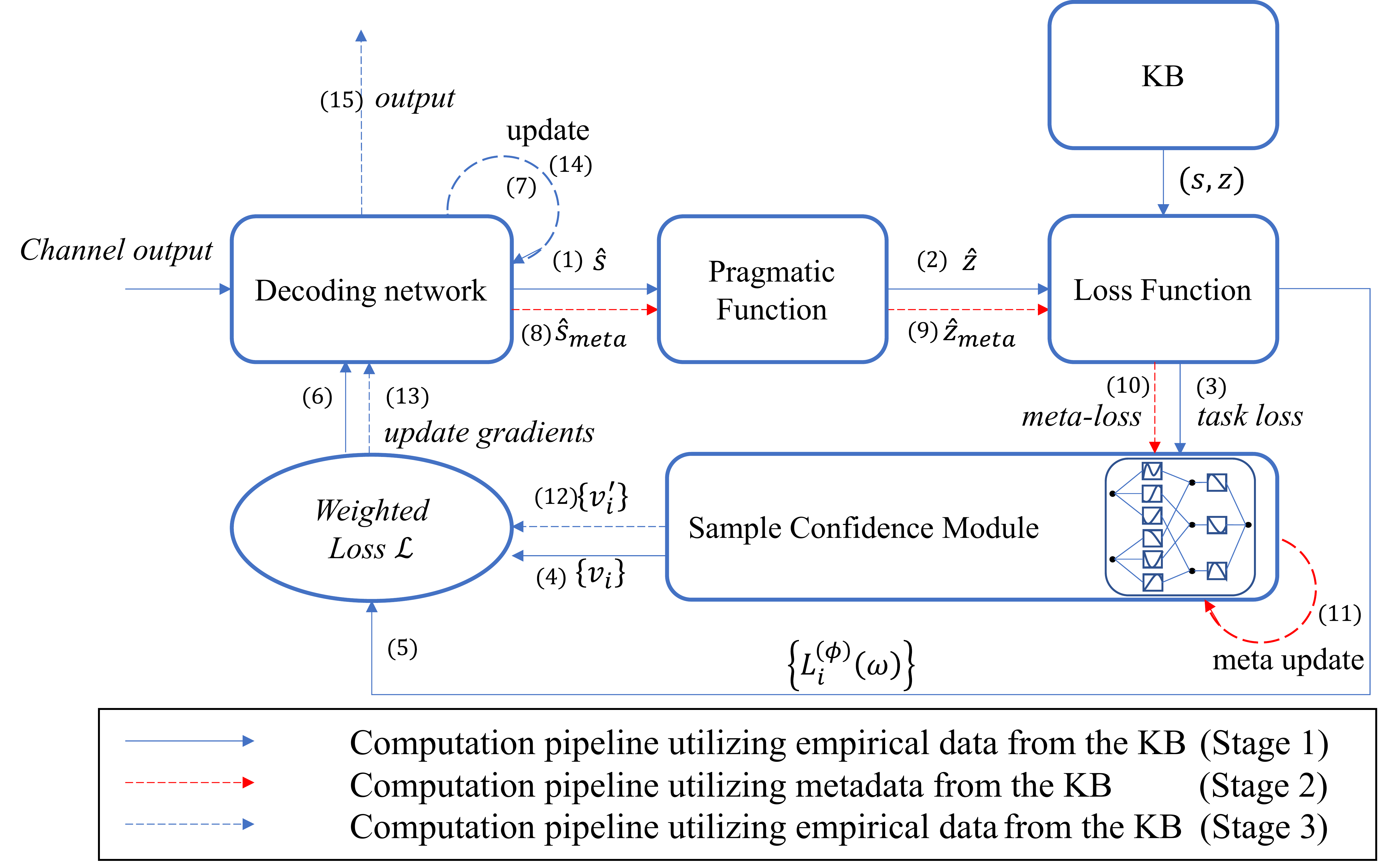}    
	\caption{Update machanism of the semantic decoding part in the task-agnostic image semantic communication system. 
For the $t$-th transmission at the receiver:
(1) Stage 1, Semantic Decoder Update: The semantic decoder is updated to $\omega_2^*(\Theta^{(t)})$ using the sample significance $v_i$ and the sample-based task loss $L_i^{(\phi)}(\omega)$. This process is illustrated by the solid blue lines in the figure.
(2) Stage 2, SCM Optimization: To optimize the SCM, metadata is leveraged to compute the meta-loss $\mathcal{L}^{\text{meta}}(\Theta^{(t)})$ based on the updated semantic decoder $\omega_2^*(\Theta^{(t)})$. The SCM parameters $\Theta^{(t)}$ are then updated to $\Theta^{(t+1)}$ according to $\mathcal{L}^{\text{meta}}(\Theta^{(t)})$. Stage 2 is represented by the dashed red lines in the figure.
(3) Stage 3, update for \emph{learners}: The updated SCM in Stage 2 recalculates the sample significance $v_i^{'}$ for $K_i$. Using the previously computed $L_i^{(\phi)}(\omega)$ from Stage 1, a new weighted loss $\mathcal{L}_t^{'}$ is generated to further refine the semantic decoder,  and the parameter of the semantic decoder is updated to $\omega_2^{(t+1)}$ based on $\mathcal{L}_t^{'}$. The necessary feedback information is then transmitted back to the transmitter for subsequent adjustments.
Stage 3 is represented by the dashed blue lines in the figure.   
Note that the red arrow indicates computations based on metadata, while the blue arrow corresponds to those based on candidate KB samples.    
}
	\label{update-machanism}
	\vspace{-8pt}
\end{figure*}

\subsection{Update Mechanism of Meta-Learning} \label{D}

The bi-level optimization problem modeled in (\ref{OP},\ref{OP-constrain}) is complex to solve. 
Therefore, inspired by model-agnostic machine learning (MAML) \cite{maml}, 
we adopt an online  meta-learning approach to update the parameters of both the \emph{learner} and \emph{meta-learner}  \cite{CMWNet}. The update mechanism of SCM and semantic coding network is shown in Fig. \ref{SCM-flowchart}.

In the $t$-th transmission within the low-level optimization process, the \emph{learner} is temporarily updated based on the \emph{meta-learner} $\mathcal{V}(\cdot \mid \Theta^{(t)})$ and sample-based \emph{task loss} $L_i^{(\phi)}(\omega)$: 
\begin{equation} \label{meta-update}
\omega^{*}(\Theta^{(t)}) = \omega^{(t)} - \alpha \frac{1}{n} \sum_{i \in M_t} \mathcal{V}_i(\Theta^{(t)}) \cdot \nabla_\omega L_i^{(\phi)}(\omega) \mid_{\omega = \omega^{(t)}},
\end{equation}
where 
$ \alpha $ is the step size, and $ n $ is the size of $ M_t $, i.e., $ n = |M_t| $. 
$\mathcal{V}_i(\Theta^{(t)})$ denotes the significance for $K_i$.

Given that updating the semantic encoder necessitates the transmission of essential information through the feedback channel, which may introduce delays due to communication resource constraints, we propose a simplified process.
During the lower-level optimization, we do not update the entire \emph{learner} $\omega=\{\omega_1, \omega_2\}$. Instead, we update only $\omega_2$ as follows:
\begin{equation} \label{meta-update2}
\omega_2^{*}(\Theta^{(t)}) = \omega_2^{(t)} - \alpha \frac{1}{n} \sum_{i \in M_t} \mathcal{V}_i(\Theta^{(t)}) \cdot \nabla_{\omega_2} L_i^{(\phi)}(\omega) \mid_{\omega = \omega^{(t)}},
\end{equation}
where $\omega_2$ is the parameters of semantic decoder, and then we have $\omega^{*}(\Theta^{(t)})=\{ \omega_1^{(t)}, \omega_2^{*}(\Theta^{(t)}) \}$.

Subsequently, in the upper-level optimization process, the \emph{meta-learner}  utilizes the updated \emph{learner} $ \omega^{*}(\Theta^{(t)}) $ to calculate the meta-loss $\mathcal{L}^{\text{meta}}(\Theta^{(t)})$ with respect to a random batch of $\mathcal{B}$, and update as follows:
\begin{equation} \label{SCM-update}
\Theta^{(t+1)} = \Theta^{(t)} - \beta \frac{1}{m} \sum_{i \in B_t} \nabla_\Theta L^{(\phi)}_i(\omega^*(\Theta)) \mid_{\Theta = \Theta^{(t)}},
\end{equation}
where $ B_t $ is the index set of the sampled batch of $\mathcal{B}$, $ \beta $ is the step size, and $ m $ is the size of $ B_t $, i.e., $ m = |B_t| $. 

Finally, the \emph{learner} $ \omega $ is updated by $ \Theta^{(t+1)} $ from a receiver-driven perspective. First, the semantic decoder is updated as follows:
\begin{equation} \label{codec-update}
\omega_2^{(t+1)} = \omega_2^{(t)} - \alpha \frac{1}{n} \sum_{i \in M_t} \mathcal{V}_i(\Theta^{(t+1)}) \cdot \nabla_{\omega_2} L_i^{(\phi)}(\omega) \mid_{\omega = \omega^{(t)}}.
\end{equation}
Then the semantic decoder calculate $\nabla_{\mathbf{y}}\mathcal{L}_t$ as follows:
\begin{equation} \label{weight-loss}
\nabla_{\mathbf{y}}\mathcal{L}_t = \frac{1}{n} \sum_{i \in M_t} \mathcal{V}_i(\Theta^{(t+1)}) \cdot \nabla_{\mathbf{y}} L_i^{(\phi)}(\omega) \mid_{\omega = \omega^{(t)}}.
\end{equation}
At last, the feedback $(\nabla_{\mathbf{y}}\mathcal{L}_t, \ \mathbf{y})$ is sent to the transmitter to update $\omega_1$ as shown in Eq. (\ref{feedback-update}).

To further investigate the update of SCM, combining Eq. (\ref{meta-update2}) and Eq. (\ref{SCM-update}), we can derive the following result using the chain rule \cite{CMWNet}: 
\begin{equation} \label{update-eq}
\Theta^{(t+1)} = \Theta^{(t)} + \frac{\alpha \beta}{n} \sum_{j \in M_t} \left(\frac{1}{m} \sum_{i \in B_t} G_{ij}\right) \frac{\partial \mathcal{V}(L_j^{(\phi)}(\omega); \Theta)}{\partial \Theta} \mid_{\Theta^{(t)}},
\end{equation}
where $ G_{ij} = {\frac{\partial L_i^{(\phi)}(\omega)}{\partial \omega}}^T \mid_{\omega^*(\Theta^{(t)})} \cdot \frac{\partial L_j^{(\phi)}(\omega)}{\partial \omega} \mid_{\omega^{(t)}} $. 
Each term in the summation indicates the gradient ascent direction for the \emph{meta-learner} $ \mathcal{V}(\cdot; \Theta) $. The coefficient $ \frac{1}{m} \sum_{i=1}^m G_{ij} $ represents the weight applied to the $ j $-th gradient term, quantifying the similarity between the gradient of a candidate KB sample and the average gradient direction derived from a mini-batch of metadata.     
A higher similarity, i.e., larger $ \frac{1}{m} \sum_{i=1}^m G_{ij} $, indicates stronger alignment with metadata, thereby resulting in a greater contribution to the increase of sample significance.
This aligns with the intuition that samples exhibiting gradient consistency with metadata are more informative for guiding \emph{learners}’ update.
The meta essence understanding for our TALSC framework is shown in Fig. \ref{meta-essence}.

\begin{algorithm}[t]
\caption{The update mechanism of SCM and semantic coding networks.}
\label{alg-A}

\begin{algorithmic}[1]
\Require  $\mathcal{K}$,  $\mathcal{B}$.
\Ensure  $\omega_1, \omega_2$, and $\Theta$.
\Statex
\hrule
\For{each transmission}

\State Take a batch of samples in $\mathcal{K}$ and $\mathcal{B}$, the emperical samples is $S$;
\State The transmitter encode and send all $\mathbf{x} = f_{\omega_1}(S)$;
\State The receiver obtains the channel output $\mathbf{y}$ and reconstructs the semantic information $\widehat{\mathbf{x}}$ according to the CSI $h$, as in Eq. (\ref{x-recover}).
\State The reciever decode data $\widehat{S} = g_{\omega_2}(\widehat{\mathbf{x}})$ and calculate the task output $\widehat{Z} = \phi(\widehat{S})$;
\State Pragmatic function calculates the sample-based \emph{task loss} $L_i^{(\phi)}(\omega)$ for the decode data, and the SCM evaluates the sample significance according to the parameters $\Theta^{(t)}$.
\State The semantic codec is temporarily updated according to Eq. (\ref{meta-update2}) to obtain the updated \emph{learner} $\omega^{*}(\Theta^{(t)})$.
\State Pragmatic function calculates the \emph{meta-loss} $\mathcal{L}^{\text{meta}}(\Theta)$ based on the meta data for $i \in B_t$, and the SCM updates $\Theta^{(t+1)}$ according to Eq. (\ref{SCM-update}).
\State The updated SCM recalculates the sample significance based on  $\Theta^{(t+1)}$ and the semantic codec updates ${\omega}^{(t+1)}$ by Eq. (\ref{codec-update}) and Eq. (\ref{feedback-update}).

\EndFor

\end{algorithmic}

\end{algorithm}

The system view of the update process is shown in Fig. \ref{update-machanism}. 
By iteratively implementing this process, the \emph{meta-learner} can automatically learn the $\mathcal{V}(\cdot | \Theta)$ guided by a small amount of metadata $\mathcal{B}$, and appropriately weight KB samples to update \emph{learner} while gradually refine $\mathcal{V}(\cdot | \Theta)$. 
Thus, even when there is a bias distribution in $\mathcal{K}$, 
the semantic coding networks can still effectively transmit semantic information for robust task execution.
The update mechanism of TALSC framework is summarized in Algorithm. \ref{alg-A}.

\section{Design of Sample Confidence Module Based on MLP and KAN} \label{SCM-KAN}

During the training of the semantic coding network, hyperparameters must be adjusted to control the significance evaluation process for better semantic coding performance. 
Given the complexity of optimizing the SEF in a bi-level optimization problem in (\ref{OP}), the MLP and KAN networks can serve as a surrogate model to approximate the behavior of $\mathcal{V}(\cdot | \Theta)$, thereby accelerating the optimization process.
In this section, we propose a construction of the SCM module based on the underlying structures of MLP and KAN.
Furthermore, leveraging the spline structure within KAN, we present the GE approach for KAN-based SCM to achieve higher resolution, enabling the capture of fine-grained details for improved significance evaluation and reducing the need for extensive model training. 
Additionally, we analyze the approximation performance and the grid dependency of the KAN-based SCM. 

\subsection{MLP and KAN Representations of SEF}

The choice of structure for SCM must consider factors such as the ability to approximate general nonlinear functions and the parameter utilization efficiency.
As a universal approximator, the MLP can theoretically approximate any continuous function with arbitrary precision, given a sufficient number of hidden neurons, as established by the universal approximation theorem \cite{UAT}. 
This implies that we can employ an MLP structure to express the $\mathcal{V}(\cdot | \Theta)$ from the pragmatic task loss to sample significance.   
By utilizing an MLP-based SCM module, the nonlinear $\mathcal{V}(\cdot | \Theta)$ can be effectively expressed, with the advantages of a relatively simple basic structure and training algorithm, making it easy to implement.

To enable the SCM to evaluate the non-linear significance effectively, each hidden node employs a ReLU activation function $\sigma$, while the output layer uses a sigmoid activation $S(x)=\frac{1}{1 + e^{-x}}$ to constrain the significance values within the range $[0,1]$. 
The SEF $\mathcal{V}(x | \Theta)$ is expressed as:
\begin{equation}
\mathcal{V}(x | \Theta) = S\left((W_{L-1} \circ \sigma \circ W_{L-2} \circ \sigma \circ \cdots \circ W_1 \circ \sigma \circ W_0)x \right),
\end{equation}
where the symbol $\circ$ denotes the composition of transformations, indicating the sequential application of $W_k$ and $\sigma$ within the SCM, and $W_k (0\leq k \leq L-1)$ denotes affine transformations.     

While the universal approximation theorem asserts that an MLP with a sufficient number of hidden neurons can approximate any continuous function to arbitrary precision, determining the exact number of neurons required for a specific approximation level remains a significant challenge in practice \cite{UAT}.   
Furthermore, the authors in \cite{mlpReLU} demonstrated that the depth of the ReLU MLP needs to reach at least $\log n_0$ to achieve the desired convergence rate for the approximation error, where $n_0$ is the number of samples. 
Moreover, while MLP in \cite{mlpReLU} focus on regression tasks, the bi-level optimization problem in (\ref{OP}) for significance evaluation introduces additional complexity, posing new challenges for constructing the $\mathcal{V}(\cdot | \Theta)$ using MLPs. 
Although increasing the width and depth of the MLP-based SCM can improve significance evaluation performance in the bi-level optimization problem (\ref{OP}), as suggested by neural scaling laws, adapting to varying task demands and precision requirements often necessitates retraining, which is both computationally expensive and resource-intensive.
Therefore, the necessity for a large number of parameters to effectively model $\mathcal{V}(x | \Theta)$ in MLP results in elevated computational costs, and increased storage demands.

In contrast, KAN-based SCM may employ a smaller architecture, as they learn general nonlinear activation functions to construct the SEF. 
KAN can construct complex functional relationships through a combination of fully connected layers and one-dimensional splines, providing an efficient and powerful approach for the SEF optimization process \cite{KAN-survey}.     

Inspired by the implement of the KAN network in \cite{KAN}, KAN representations of the SCM can be described as follows.

First, the shape of SCM network can be represented by a set of integers $[n_0 = d, n_1, n_2, \ldots, n_l, \ldots, n_L = 1]$,
where $ n_l $ denotes the number of nodes in the $ l $-th layer.

Second, we denote the value of the $ i $-th node at the $ l $-th layer as $ x_{l,i} $. Between the $ l $-th layer and the $ (l+1) $-th layer, there are $ n_l n_{l+1} $ activation functions, where the activation function from $ x_{l,j} $ to $ x_{l+1,i} $ is represented as $ \phi_{l,i,j} $. Each activation function contains trainable parameters. Thus, the value of $ x_{l+1,i} $ at the $ (l+1) $-th layer can be expressed as $(i = 1, 2, ..., n_{l+1})$:
\begin{equation}
x_{l+1,i} = \sum_{j=1}^{n_l} \phi_{l,i,j}(x_{l,j}),
\end{equation}
where $ \phi_{l,i,j}(\cdot) $ is parameterized by the B-spline function, and we can express the computations at each layer in the matrix form as follows:    
\begin{equation}
\bold{x_{l+1}} =
\begin{pmatrix}
\phi_{l,1,1}(\cdot) & \phi_{l,1,2}(\cdot) & \cdots & \phi_{l,1,n_l}(\cdot) \\
\phi_{l,2,1}(\cdot) & \phi_{l,2,2}(\cdot) & \cdots & \phi_{l,2,n_l}(\cdot) \\
\vdots & \vdots & \ddots & \vdots \\
\phi_{l,n_{l+1},1}(\cdot) & \phi_{l,n_{l+1},2}(\cdot) & \cdots & \phi_{l,n_{l+1},n_l}(\cdot)
\end{pmatrix}
\bold{x_l},
\end{equation}
where $\bold{x_l} = [x_{l,1}, x_{l,2}, ..., x_{l,n_l}]^T.$
This can be compactly written as:
\begin{equation}
\bold{x_{l+1}} = \Phi_l \circ \bold{x_l},
\end{equation}
where $ \Phi_l = \{\phi_{l,i,j}\}^{n_{l+1} \times n_l} $ is defined as the activation function matrix from the $ l $-th layer to the $ (l+1) $-th layer, with the dimensions $ n_{l+1} \times n_l $. 

Then, the SCM network can be represented as follows:
\begin{equation}
\mathcal{V}(x | \Theta) = S\left( (\Phi_{L-1} \circ \Phi_{L-2} \circ \cdots \circ \Phi_1 \circ \Phi_0)x \right).
\end{equation}

The authors in \cite{KAN} have shown that KAN effectively mitigates the risk of overfitting and exhibits a better neural scaling rate. 
In our TALSC framework, the SCM incorporates the KAN structures, which leverages spline-based activation functions that allows for finer granularity and greater flexibility in capturing complex relationships for significance scores.
Therefore, the KAN-based SCM provides a more efficient and accurate representation of the sample significance levels than the MLP-based SCM, which is crucial for optimizing the meta-loss.

\subsection{Grid Dependency and Approximation Theorem of KAN-based SCM}

In KANs, basis functions can be classified based on their dependency on the grid. 
Static basis functions, such as Jacobi polynomials \cite{Jacobi1, Jacobi2}, are independent of the grid, which restrict their adaptability to changes in grid resolution or training conditions. 
As highlighted by the authors in \cite{grid-depend}, maintaining grid dependency for KAN basis functions is crucial, as it enables adaptive training and allows the model's representations to dynamically adjust in response to grid updates during training.
For KAN-based SCM, the grid dependency and approximation error bounds are considered as follows. 
We assume that the target SEF $\mathcal{V}(x) \in [0,1]$ for the optimization problem (\ref{OP}) can be represented by the Kolmogorov-Arnold theorem as follows: 
\begin{equation}  \label{vform}
\mathcal{V}(x)=S\left(\left(\Phi_{L - 1}\circ\Phi_{L - 2}\circ\cdots\circ\Phi_{1}\circ\Phi_{0}\right)x\right),
\end{equation}
where $S(x)$ represents the sigmoid function and each element $\phi_{l,i,j}$ in $\Phi_l$ ($0\leq l\leq L - 1$) is $(k+1)$ times continuously differentiable. 
Due to the spline structure in our KAN-based SCM, it is dependent on a finite grid point set. 
To express this dependency, we denote the nonlinear transformation from the $ l $-th layer to the $(l+1)$-th layer as $ \Phi_l^G $, which is approximated by using a spline with a grid size $ G $.
Thus, we have the approximation error bound for the KAN-based SCM in Theorem \ref{theorem1}.
\begin{theorem} \label{theorem1}
Suppose the  SEF $\mathcal{V}(x)\in [0,1]$ of an SCM can be expressed as shown in Eq. (\ref{vform}),     
Then, there exists a meta-learner constructed by the $k$-th order B-spline function $\Phi_l^G$ ($0\leq l\leq L - 1$):
\begin{equation} \label{meta-representation}
\mathcal{V}(x; \Theta) = S\left(\left(\Phi_{L - 1}^{G}\circ\Phi_{L - 2}^{G}\circ\cdots\circ\Phi_{0}^{G}\right)x\right). 
\end{equation}
Then, there exists a constant $C$, which depends on $\mathcal{V}(x)$, 
such that for any $0\leq m\leq k$, we have the approximation bound in terms of  $G$:
\begin{equation} \label{approximate-bound}
\|\mathcal{V}(x)-\mathcal{V}(x; \Theta)\|_{C^m}\leq CG^{-k - 1 + m},
\end{equation}
where $C^m$ measures the magnitude of derivatives up to order $m$ as follows:
\[\|g\|_{C^m}=\max_{|\beta|\leq m}\sup_{x\in [0,1]^n}|D^{\beta}g(x)|.\]
\end{theorem}
\begin{proof}
For $S(x)=\frac{1}{1 + e^{-x}}$, we can directly have $S'(x)\leq\frac{1}{4}$. 
Since the derivative of $S(x)$ is bounded, it can be concluded that $S(x)=\frac{1}{1 + e^{-x}}$ is a Lipschitz continuous function. 
Thus, we have: 
\begin{equation} \label{sigmoid-scale}
\vert S(x_1)-S(x_2)\vert\leq\frac{1}{4}\vert x_1 - x_2\vert.
\end{equation}

On the other hand, according to the classical one-dimensional B-spline theory \cite{bspline}, and the fact that $\phi_{l,i,j}$ is a uniformly bounded continuous function, there exists a finite grid B-spline function $\phi_{l,i,j}^G$ for $0\leq m\leq k$ we have,
\begin{equation}
\begin{aligned}
&\|\left(\phi_{l,i,j}\circ\Phi_{l - 1}\circ\Phi_{l - 2}\circ\cdots\circ\Phi_{1}\circ\Phi_{0}\right)x\\
&-\left(\phi_{l,i,j}^{G}\circ\Phi_{l - 1}\circ\Phi_{l - 2}\circ\cdots\circ\Phi_{1}\circ\Phi_{0}\right)x\|_{C^m}\\
&{\leq} CG^{-k - 1 + m},
\end{aligned}
\end{equation}
where $C$ is independent of $G$.

Therefore, we define the residual $R_l$ as
\begin{align*}
&R_l:=\left(\Phi_{L - 1}^{G}\circ\cdots\circ\Phi_{l + 1}^{G}\circ\Phi_{l}\circ\Phi_{l - 1}\circ\Phi_{l - 2}\circ\cdots\circ\Phi_{1}\circ\Phi_{0}\right)x\\
&-\left(\Phi_{L - 1}^{G}\circ\cdots\circ\Phi_{l + 1}^{G}\circ\Phi_{l}^{G}\circ\Phi_{l - 1}\circ\Phi_{l - 2}\circ\cdots\circ\Phi_{1}\circ\Phi_{0}\right)x,
\end{align*}
and $R_l$ satisfies: 
\begin{equation} \label{Rl-scale}
\|R_l\|_{C^m}\leq CG^{-k - 1 + m}.
\end{equation}

Thus, there exists a constant $C$ which depends on $\mathcal{V}(x)$ and its representation $\mathcal{V}(x; \Theta)$, we have 
\begin{align*}
&\| \mathcal{V}(x)-\mathcal{V}(x; \Theta) \|_{C^m} \\ 
&\overset{(a)}{\leq} \frac{1}{4} \|\left(\Phi_{L - 1} \circ \Phi_{L - 2} \circ \cdots \circ \Phi_{0}\right)x  \\ &- \left(\Phi_{L - 1}^{G} \circ \Phi_{L - 2}^{G} \circ \cdots \circ \Phi_{0}^{G}\right)x \|_{C^m}  \\ 
&\overset{(b)}{\leq} \frac{1}{4} \| R_{L - 1} + R_{L - 2} + \cdots + R_{0} \|_{C^m} \\ 
&\overset{(c)}{\leq} \frac{1}{4} \left( \| R_{L - 1} \| + \| R_{L - 2} \| + \cdots + \| R_{0} \| \right)  \\ 
&\overset{(d)}{\leq} C G^{-k - 1 + m} ,
\end{align*}
where $(a)$ follows from Eq. (\ref{sigmoid-scale}); $(b)$ is due to the definition of $ R_l $; $(c)$ uses the triangle inequality; $(d)$ follows from Eq. (\ref{Rl-scale}), which establishes a bound on the residuals.

\end{proof}

\subsection{Grid Extension Approach for KAN-based SCM}

The KAN-based SCM leverages the flexibility of spline functions, enabling adaptable approximation of $\mathcal{V}(\cdot | \Theta)$ through adjustable grid resolutions. 
To derive a fine-grained $\mathcal{V}(\cdot | \Theta)$ for significance evaluation, the SCM is initially trained with a compact parameter set and subsequently expanded into a higher-resolution model by refining its grid, eliminating the need for complete retraining.

Specifically, in the KAN-based SCM, the parameters $\Theta$ are composed of the B-spline parameters that construct the activation functions. 
These parameters can be grouped according to their corresponding activation functions, such that $\Theta=\{\theta_1, \theta_2, …, \theta_q, … \theta_{N_0} \}$, where $N_0$ is the total number of activation functions, and $\theta_q$ is the parameters for the $q$-th activation function in SCM. 
For a general KAN structure with a shape defined as $[n_0, n_1, n_2, \ldots, n_l, \ldots, n_L]$ and $N_0=\sum_{i=0}^{L-1}n_ln_{l+1}$. 
Given the grid size $G$ and spline order $k$, each activation function $f(\cdot; \theta_q) (q=1, \ldots, N_0)$ is defined over the interval $[a, b]$, which is uniformly divided into $G$ subintervals, denoted as $\{ t_0 = a, t_1, t_2, \ldots, t_G = b \}$. 
To construct a spline of order $k$ for representing the activation function, the knot vector is expanded to $\{ t_{-k}, \ldots, t_{-1}, t_0, \ldots, t_G, t_{G+1}, \ldots, t_{G+k} \}$. 
This setup results in $(G + k)$ B-spline basis functions, and the $i$-th  B-spline basis function $B_{i}(x)$ is nonzero only within the interval $[t_{-k+i-1}, t_{i}]$, where $i = 1, \ldots, G+k$. 
Thus, the activate function $f(\cdot; \theta_q)$ can be expressed as follows:
\begin{equation}
f(x; \theta_q) = \sum_{i = 1}^{G + k} c_{qi} B_{i}(x), 
\end{equation}
where $c_{qi}$ is the corresponding coefficient for $B_{i}(x)$ in the $q$-th activate function within SCM, and $\theta_q = \{c_{qi}\}_{i=1}^{G+k}$.

Given the initial SCM with the grid size $G_1$ and spline order $k_1$, we have the SCM parameter $\Theta^{(1)}=\{\theta_1^{(1)}, \theta_2^{(1)}, …, \theta_q^{(1)}, … \theta_{N_0}^{(1)} \}$, 
and the $q$-th activate function can be expressed as follows:
\begin{equation}
f^{(1)}(x; \theta_q^{(1)}) = \sum_{i = 1}^{G_1 + k_1} c_{qi}^{(1)} B_{i}^{(1)} (x), 
\end{equation}
where $B_{i}^{(1)}(x)$ denotes the B-spline basis functions for grid size $G_1$ and spline order $k_1$, $c_{qi}^{(1)}$ is the corresponding coefficient for $B_{i}^{(1)}(x)$, and $\theta_q^{(1)} = \{c_{qi}^{(1)}\}_{i=1}^{G_1+k_1}$. 

Given a finer grid with $G_2$ and spline order $k_2$, we have the SCM parameter $\Theta^{(2)}=\{\theta_1^{(2)}, \theta_2^{(2)}, …, \theta_q^{(2)}, … \theta_{N_0}^{(2)} \}$, 
and the $q$-th activate function can be expressed as follows:
\begin{equation}
f^{(2)}(x; \theta_q^{(2)}) = \sum_{i = 1}^{G_2 + k_2} c_{qi}^{(2)} B_{i}^{(2)} (x), 
\end{equation}
where $ \theta_q^{(2)}=\{c_{qi}^{(2)}\}_{i=1}^{G_2+k_2} $ can be initialized by minimizing the MSE between $ f^{(1)}(x; \theta_q^{(1)}) $ and $ f^{(2)}(x; \theta_q^{(2)}) $ under the distribution $p(x)$ for $x$, and we formulate this minimization problem as follows:
\begin{equation} \label{GE-OP}
\begin{split}
\theta_q^{(2)*} = 
\underset{\theta_q^{(2)}}{\text{argmin}} \ 
\mathbb{E}_{x \sim p(x)} \bigg( 
& \sum_{j = 1}^{G_2 + k_2} c_{qj}^{(2)} B_j^{(2)}(x) \\
& - \sum_{i = 1}^{G_1 + k_1} c_{qi}^{(1)} B_i^{(1)}(x)
\bigg)^{2}.
\end{split}
\end{equation}
This can be solved using least squares methods.      

First, we sample data points $x_i$ from $p(x)$ and obtain $ \{ x_i \}_{i=1}^n $.
Then, the basis spline matrix for grid $ G_1 $ and spline order $ k_1 $ is defined as:
\begin{equation} \label{basis-spline-matrix1}
\mathcal{B}_1 = \begin{bmatrix}
B_1^{(1)}(x_1) & B_2^{(1)}(x_1) & \cdots & B_m^{(1)}(x_1) \\
B_1^{(1)}(x_2) & B_2^{(1)}(x_2) & \cdots & B_m^{(1)}(x_2) \\
\vdots & \vdots & \ddots & \vdots \\
B_1^{(1)}(x_n) & B_2^{(1)}(x_n) & \cdots & B_m^{(1)}(x_n)
\end{bmatrix},
\end{equation}
where $ m = G_1 + k_1 $. 
Similarly, the basis spline matrix for grid $ G_2 $ and spline order $ k_2 $ is defined as:
\begin{equation} \label{basis-spline-matrix2}
\mathcal{B}_2 = \begin{bmatrix}
B_1^{(2)}(x_1) & B_2^{(2)}(x_1) & \cdots & B_{m'}^{(2)}(x_1) \\
B_1^{(2)}(x_2) & B_2^{(2)}(x_2) & \cdots & B_{m'}^{(2)}(x_2) \\
\vdots & \vdots & \ddots & \vdots \\
B_1^{(2)}(x_n) & B_2^{(2)}(x_n) & \cdots & B_{m'}^{(2)}(x_n)
\end{bmatrix},
\end{equation}
where $ m' = G_2 + k_2 $. We can then express the minimization problem (\ref{GE-OP}) in the matrix form as follows:    
\begin{equation} \label{e34}
\theta_q^{(2)*} = \underset{\theta_q^{(2)}}{\text{argmin}} \ (\mathcal{B}_1 \theta_q^{(1)} - \mathcal{B}_2 \theta_q^{(2)})^T (\mathcal{B}_1 \theta_q^{(1)} - \mathcal{B}_2 \theta_q^{(2)}).
\end{equation}
Eq. (\ref{e34}) 
leads to the analytical solution for the optimal initialization $ \theta_q^{(2)*} $:
\begin{equation} \label{ge}
\theta_q^{(2)*} = (\mathcal{B}_2^T \mathcal{B}_2)^{-1} \mathcal{B}_2^T \mathcal{B}_1 \theta_q^{(1)} = \mathcal{B}_2^+ \mathcal{B}_1 \theta_q^{(1)},
\end{equation}
where $ \mathcal{B}_2^+ $ denotes the pseudoinverse of $ \mathcal{B}_2 $. We define the transformation matrix $ \mathbf{T}_{12} = \mathcal{B}_2^+ \mathcal{B}_1 $ as the grid transformation matrix from $ (G_1, k_1) $ to $ (G_2, k_2) $.

Utilizing the transformation matrix $ \mathbf{T}_{12} $ enables the independent expansion of all splines in the SCM, resulting in finer granularity for SEFs. 
By applying the transformation $\theta_q^{(2)}=\mathbf{T}_{12}\theta_q^{(1)}$ for $1\leq q \leq N_0$, we can derive the new set of SCM parameters $\Theta^{(2)}$ from the original parameters $\Theta^{(1)}$. 
After this, minimal fine-tuning is performed to ensure that the refined parameters $\Theta^{(2)}$ are optimized for the new grid resolution.
The GE approach for KAN-based SCM is summerized in Algorithm \ref{GE}.
Indeed, increasing  $G$ is typically employed to refine the SCM rather than indefinitely increasing $k$. 
This is because excessively large $k$ may introduce excessive oscillations, leading to optimization challenges and potential instability in the learning process.     

\begin{algorithm}[t]
\caption{SCM Grid Extension Approach.}
\label{GE}
\begin{algorithmic}[1]
\Require $ \Theta^{(1)} $, $ G_1 $, target SEF approximation tolerance $ \epsilon $.
\Ensure  $ \Theta^{(2)} $.
\Statex
\hrule
\State Infer the target SCM grid size $ G_2 $ based on the approximate bound in Eq.~(\ref{approximate-bound}) and the tolerance $ \epsilon $. If $ G_2 \leq G_1 $, terminate; otherwise, proceed to the next step.
\State Construct the basis spline matrices $ \mathcal{B}_1 $ and $ \mathcal{B}_2 $ by Eq.(\ref{basis-spline-matrix1}) and Eq. (\ref{basis-spline-matrix2}) for grid sizes $ G_1 $ and $ G_2 $, respectively.
\State Compute the grid transformation matrix $ \mathbf{T}_{12} $ as defined in Eq.~(\ref{ge}).
\For{$q=1:N_0$}
\State $\theta_q^{(2)}=\mathbf{T}_{12}\theta_q^{(1)}$.
\EndFor 
\State $\Theta^{(2)}=\{\theta_1^{(2)}, \theta_2^{(2)}, …, \theta_{N_0}^{(2)} \}$
\State Fine-tune $ \Theta^{(2)} $ to ensure the SEF fits the finer grid appropriately and obtain the final fine grid SCM parameters $ \tilde{\Theta}^{(2)} $.
\end{algorithmic}
\end{algorithm}

\section{Simulation Results and Discussions} \label{sim}

\subsection{Simulation Setup}

In this section, we compare the proposed TALSC framework with the task-agnostic deep joint source-channel coding (TA-DeepJSCC) approaches under AWGN, Rayleigh fading, and Rician fading channels \cite{task-unaware}.
The simulated dataset is CIFAR-10, comprising 60,000 images with dimensions of $32\times32$ pixels, categorized into 10 classes.
Specifically, 50,000 images are allocated for training, and the remaining 10,000 are designated for testing.
The pragmatic task at the receiver is image classification. 
Therefore, we use the accuracy of classification and recognition at the receiver as the accuracy of our semantic transmission recovery, i.e. semantic recovery accuracy (SRA).
In addition, we use MS-SSIM \cite{mssim} to evaluate the image quality of the reconstructed images.
The semantic codec utilizes a CNN architecture. It consists of an encoder with four layers of $6\times6$ convolutional layers, each followed by ReLU activations and specific pooling operations. The decoder also comprises four layers, employs transposed convolutional layers to reconstruct the output based on the received signal $\mathbf{y}$. The decoder aims to incrementally restore the spatial resolution of the feature maps until they closely match the original image dimensions.

\subsection{Results}

\subsubsection{KB with label flipping noise}

Semantic label flipping noise has a serious impact on the coding process and can mislead semantic coding networks to produce incorrect results, thereby leading to task failure \cite{label-noise, label-noise2}.
To quantify the degree of label flipping noise, we define the Flipping Noise Rate (FNR), which represents the probability that labels in the KB are flipped to another category due to the influence of semantic noise.
We can construct a flipping pattern matrix with FNR = $p$ as follows: 
\[
\begin{bmatrix}
(1-p) & p/(n - 1) &\cdots & p/(n - 1)\\
p/(n - 1) & (1-p) &\cdots & p/(n - 1)\\
\vdots &\vdots &\ddots &\vdots\\
p/(n - 1) &\cdots & p/(n - 1) & (1-p)
\end{bmatrix}, 
\]
where $n$ denotes the number of label categories. In this $n$-by-$n$ matrix, the diagonal elements are $(1 - p)$, and each off-diagonal element is $p/(n - 1)$. Consequently, the sum of the elements in each row is equal to 1.     
For the $i$-th category, sampling is performed according to the data distribution of the $i$-th row in the flipping pattern matrix. 
Specifically, there is a probability of $(1 - p)$ that the label remains unchanged, while with probability $p$, it randomly flips to any of the incorrect categories. 
This flipping probability $p$ is the FNR defined in this paper. 
A higher FNR indicates a greater proportion of corrupted samples in KB.

\begin{figure}[t]
	\centering
	\includegraphics[width=0.5\textwidth]{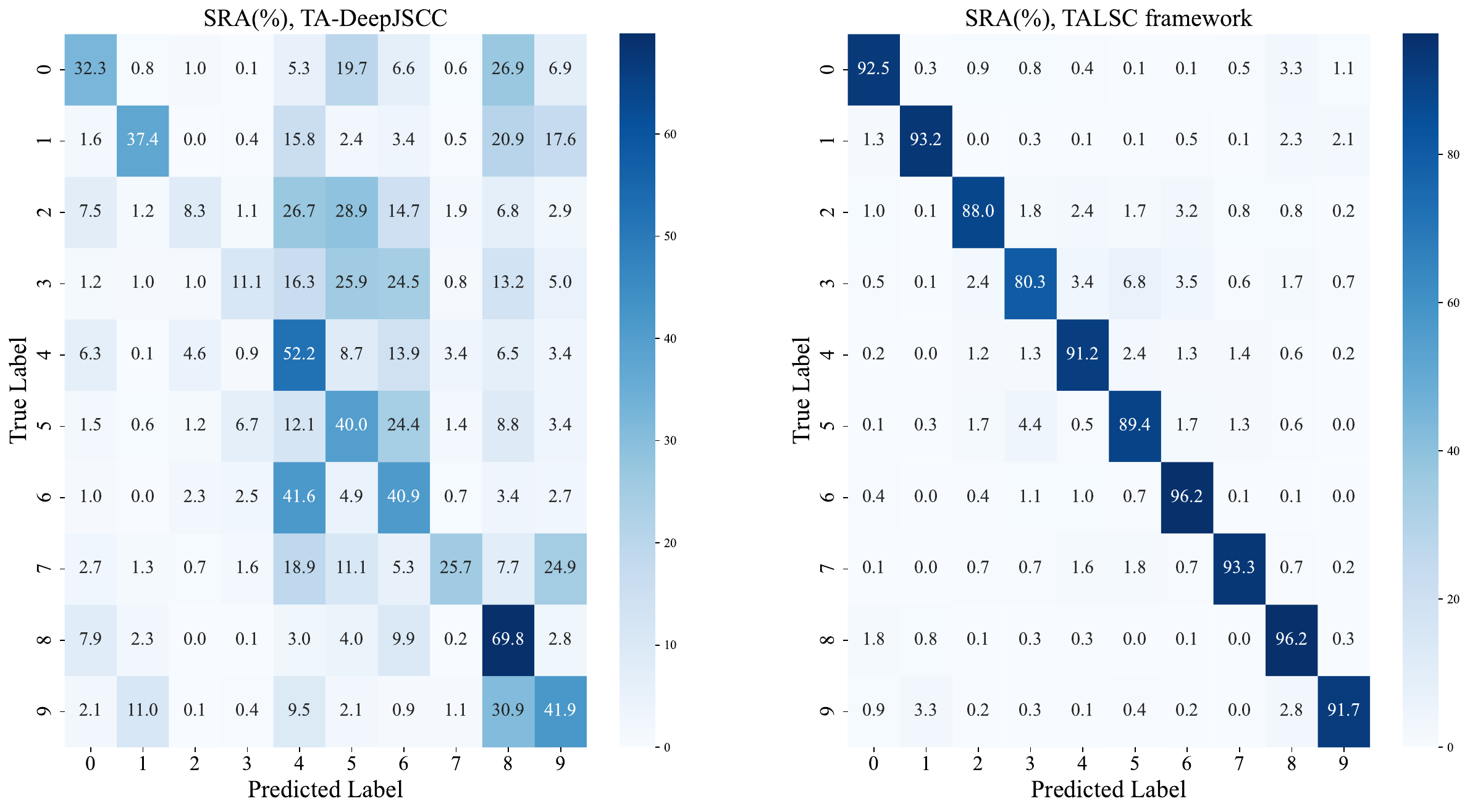}
	\caption{Confusion matrices of TA-DeepJSCC and TALSC framework over AWGN channel, SNR=8 dB.}
	\label{cf-mat}
	\vspace{-8pt}
\end{figure}

Fig. \ref{cf-mat} shows the confusion matrices for TA-DeepJSCC and TALSC framework, qualitatively illustrating the impact of label flipping noise on the models. 
The test conditions include an SNR of 8 dB over AWGN channel, and FNR is 0.4.     
Semantic codecs are shown to be highly susceptible to label flipping noise, as evidenced by the low accuracy (around 10\%) for categories 2 and 3, indicating that the semantic codec struggles to extract their semantic information. 
In contrast, the TALSC framework, aided by the SCM, effectively suppresses low-significance noise samples, thereby enhancing the robustness of semantic transmission.

\begin{figure}[t]
    \subfloat[AWGN channel.]{
        \includegraphics[width=0.85\columnwidth]{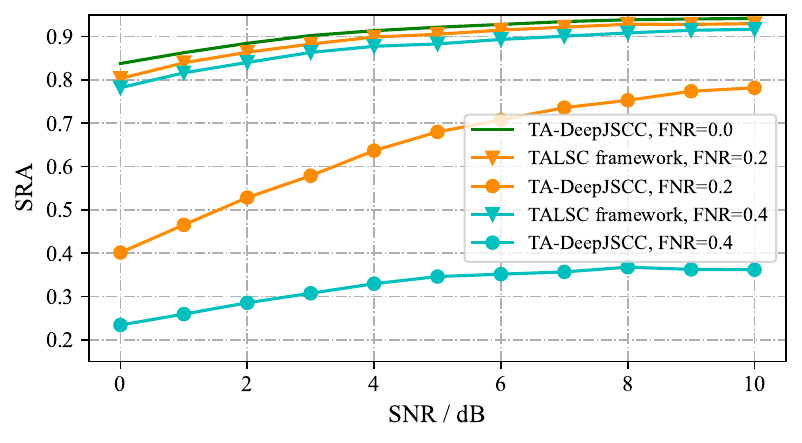}
        \label{fig:awgn}
    }
    \vspace{-8pt}
    \subfloat[Rician channel, Rician factor $K=8$.]{
        \includegraphics[width=0.85\columnwidth]{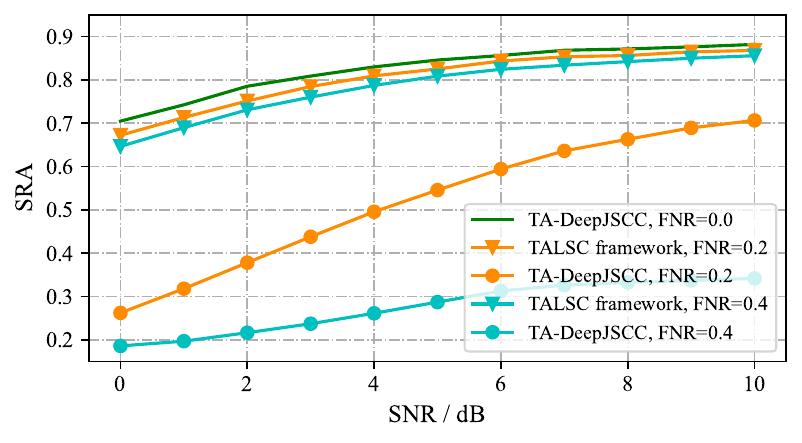}
        \label{fig:rician}
    }
    \vspace{-8pt}
    \subfloat[Rayleigh channel.]{
        \includegraphics[width=0.85\columnwidth]{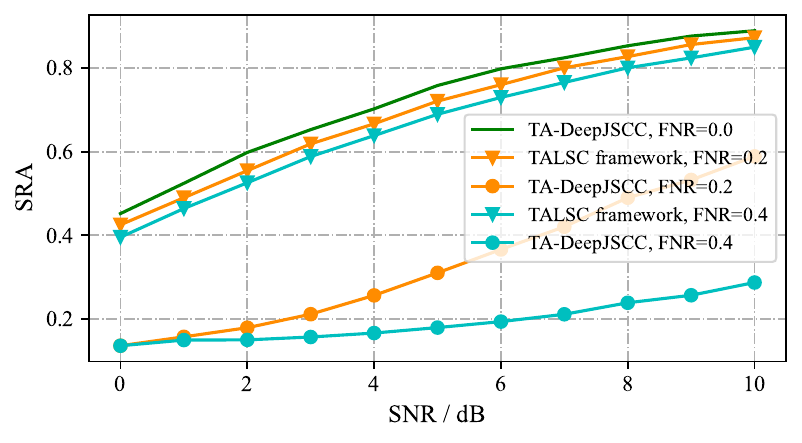}
        \label{fig:rayleigh}
    }
    \caption{Task performance under AWGN and fading channels with flipping noise in KB, our TALSC framework significantly outperforms the TA-DeepJSCC with FNR= 0.2 and 0.4.}
    \label{sra}
    \vspace{-8pt}
\end{figure}

We further evaluated the performance of the TALSC framework over AWGN and fading channels. 
As shown in Fig. \ref{sra}, our TALSC framework shows stronger robustness than TA-DeepJSCC. 
As the FNR increases, the SRA of DeepJSCC is significantly affected. 
However, the TALSC framework can still maintain a relatively high SRA and perform closely  when FNR = 0.

\begin{figure}[t]
	\includegraphics[width=0.85\columnwidth]{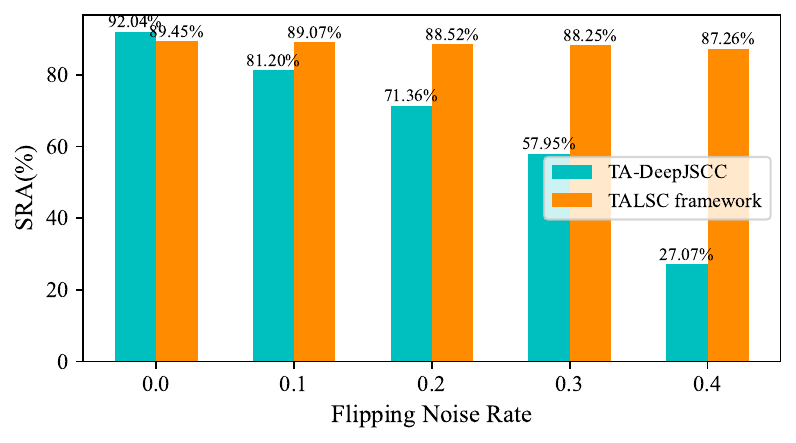}
	\caption{The SRA at different FNRs.}    
	\label{nr}
	\vspace{-10pt}
\end{figure}

The simulation results of different FNRs in the KB on the semantic coding network is shown in Fig. \ref{nr}.
We compare the performance of TA-DeepJSCC  and our TALSC, and it can be observed that our TALSC framework performs slightly worse than the baseline model only when the FNR=0, indicating no flipping noise.
In this case, our TALSC framework achieves an accuracy of 89.45\%, while the baseline achieves an accuracy of 92.04\%.
However, as the FNR increases from 0 to 0.4, the accuracy of TA-DeepJSCC drops by 64.97\%, while TALSC framework only experiences a decrease of 2.19\%.

\subsubsection{KB with imbalanced data}

In practical applications, the data collected is typically imbalanced, with some categories having significantly more samples than others \cite{imb}.  
However, for some minority categories where data is hard to obtain, they contain information that cannot be overlooked. 
In semantic transmission, semantic codecs often inadequately capture the semantic features of minority categories, resulting in suboptimal performance for data from these classes in downstream tasks. 
To address this, we simulate a class-imbalanced dataset based on CIFAR-10 and enhance the representation of minority categories in semantic codecs.

We constructs a data distribution for $n$ classes according to the imbalance factor $f$. Assuming that the class with the largest number of samples contains $N$ data samples, the number distribution of each category is constructed as follows (in descending order):
\[
\left[N,\frac{N}{f^{\frac{1}{n - 1}}},\frac{N}{f^{\frac{2}{n - 1}}},\cdots,\frac{N}{f^{\frac{k}{n - 1}}},\cdots,\frac{N}{f}\right].
\]

\begin{figure}[t]
    \centering
    \subfloat[The F1-score at different SNRs.]{
        \includegraphics[width=0.85\columnwidth]{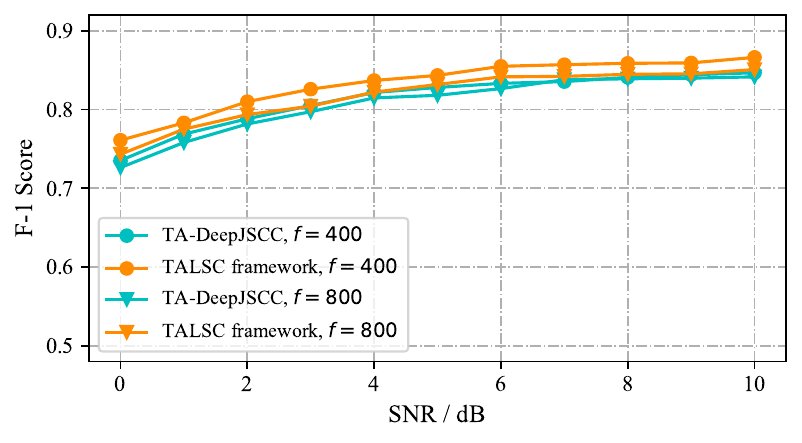}
        \label{fig:awgn_f1}
    }  
    \vspace{-8pt}
    \subfloat[The MS-SSIM for underrepresented class at different SNRs.]{
        \includegraphics[width=0.85\columnwidth]{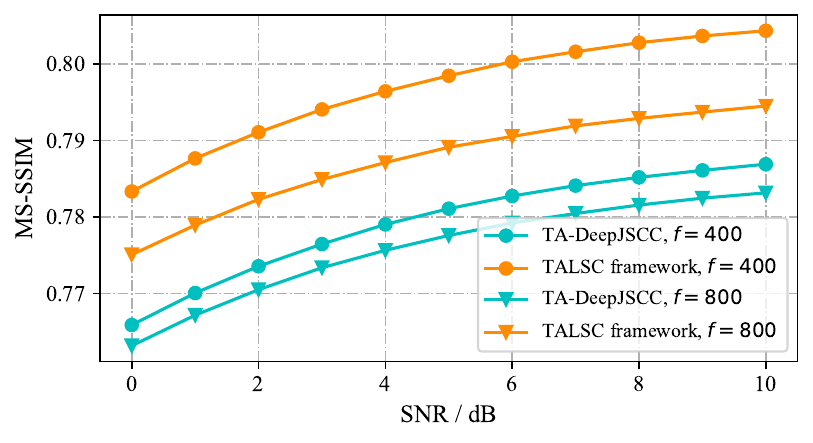}
        \label{fig:awgn_ssim}
    }
    
    \caption{Performance under AWGN channel with imbalanced KB, our TALSC framework has slightly higher F-1 score and significantly better MS-SSIM compare to the TA-DeepJSCC.}
    \label{awgn_imb}
    \vspace{-8pt}
\end{figure}

\begin{figure}[t]
	\centering
	\includegraphics[width=0.85\columnwidth]{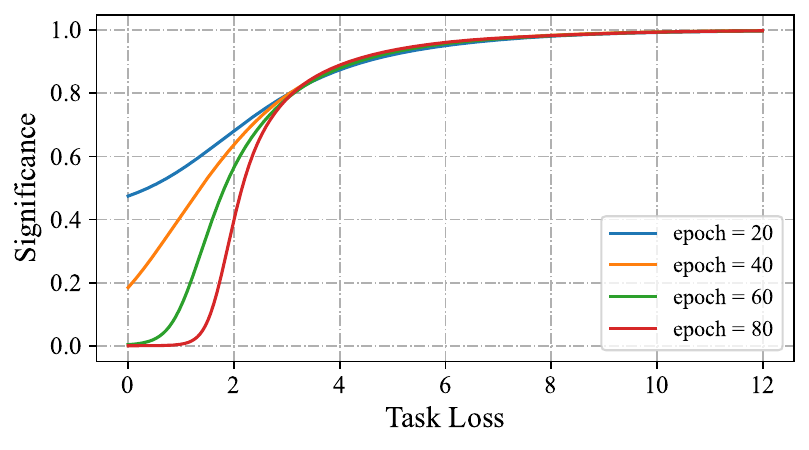}
	\caption{The significance evaluation function evolution under imbalanced KB.}
	\label{evol-imb}
	\vspace{-10pt}
\end{figure}

Fig. \ref{awgn_imb} presents a comparison of the performance of the TALSC framework and TA-DeepJSCC with $f=400$ and $f=800$, respectively. 
To evaluate the transmission capability of the semantic coding networks for minority categories in KB, we adopt the F1-score metric. 
Specifically, we assess the semantic representation capability of the semantic coding networks for minority categories by calculating the F1-score in the receiver’s pragmatic task. 
Simulation results indicate that the TALSC framework effectively enhances the representation and transmission quality of minority categories.

\subsubsection{Evolution of SEF in SCM}

The update of SCM in Eq. (\ref{update-eq}) indicates that the sample significance that is more closely consistent with the meta-knowledge will be strengthened, while the sample weights that deviate from such meta-knowledge will be suppressed. 
In this part, we observe this phenomenon by analyzing the update process in the case of class imbalance and label flipping noise and provide a reasonable explanation.

\begin{figure*}[t]
    \centering
    \subfloat[]{ 
        \includegraphics[width=0.48\linewidth]{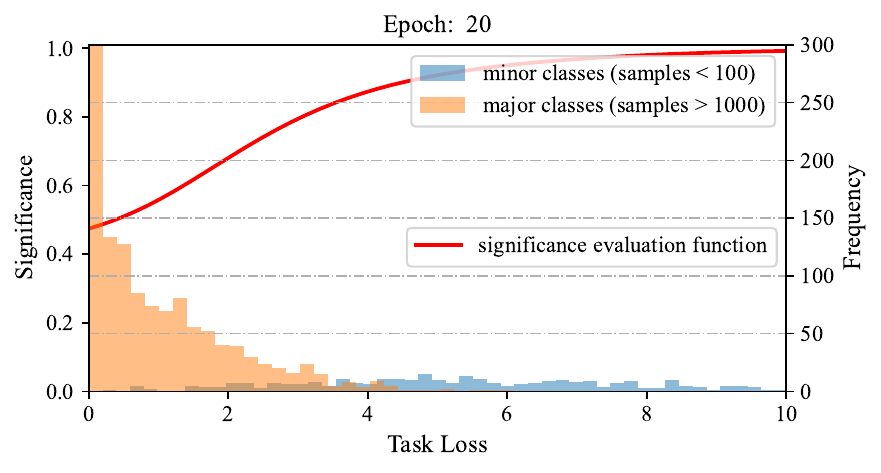}
        \label{fig:imb20}
    }
    \subfloat[]{
        \includegraphics[width=0.48\linewidth]{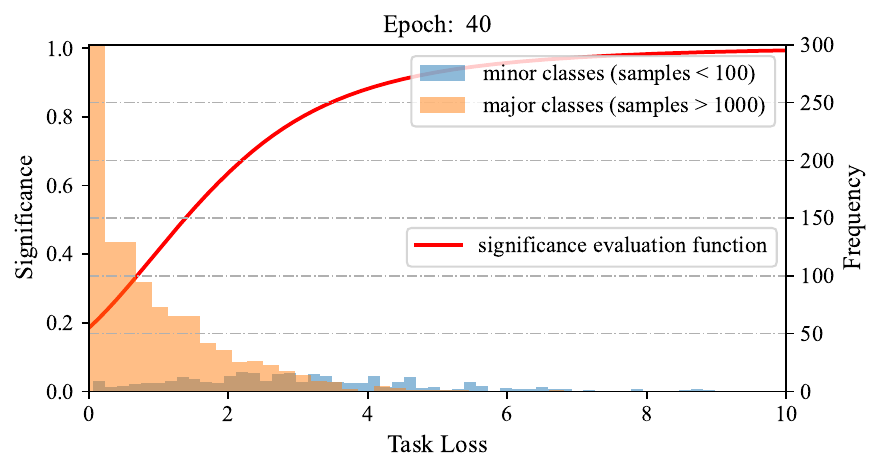}
        \label{fig:imb40}
    } 
    \vspace{-8pt}
    \subfloat[]{
        \includegraphics[width=0.48\linewidth]{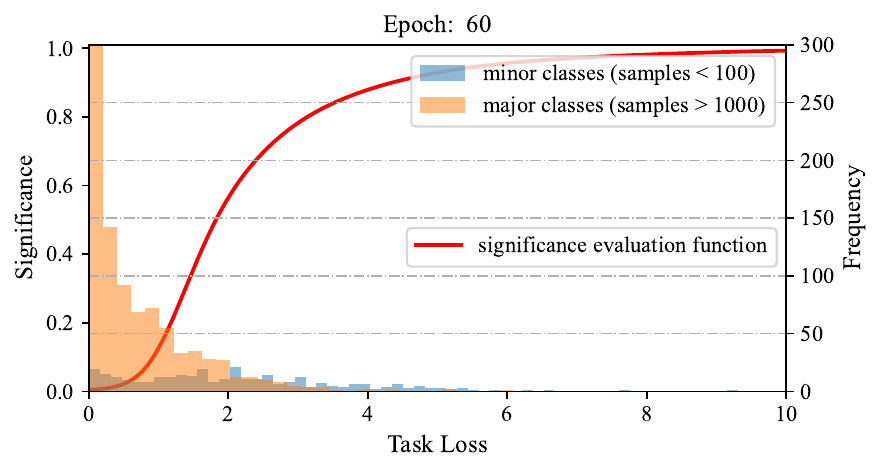}
        \label{fig:imb60}
    }
    \subfloat[]{
        \includegraphics[width=0.48\linewidth]{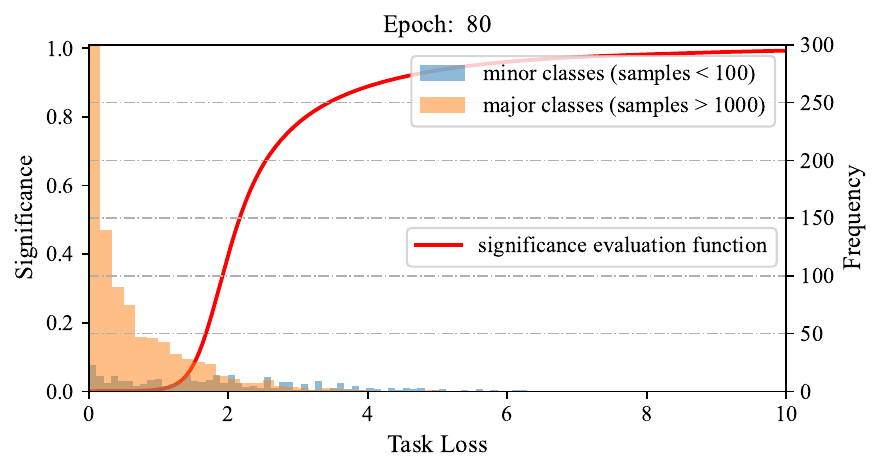}
        \label{fig:imb80}
    }
    \caption{
    Loss distribution across different training stages under class imbalance. Subfigures (a)–(d) illustrate the evolution of loss distribution alongside the corresponding progression of significance evaluation function.   
    }
    \label{loss-dist-imb}
    \vspace{-8pt}
\end{figure*}

Fig. \ref{evol-imb} shows the evolution of SEF under class imbalance with an imbalance factor $f=400$. 
The significance of samples with large losses increases over time, indicating that the SCM applies a weighting strategy that prioritizes samples with higher task losses. 
These samples are typically from minority categories that have been insufficiently trained. 
By assigning greater significance to such samples, the semantic codec is compelled to learn more extracting semantic information from underrepresented categories, thereby enhancing performance for these classes in downstream tasks.

We further present the loss distribution of the pragmatic task at different stages under class imbalance as shown in Fig. \ref{loss-dist-imb}. 
The goal is for the data from minority categories to be fully trained, meaning that the loss distribution of minority category samples should closely resemble that of majority category samples. 
We can observe that, with the TALSC framework, the monotonically increasing SEF causes the loss distribution of minority category samples to gradually converge toward that of the majority category as training progresses from Figs. \ref{fig:imb20} to \ref{fig:imb80}.    

\begin{figure}[t]
	\centering
	\includegraphics[width=0.9\columnwidth]{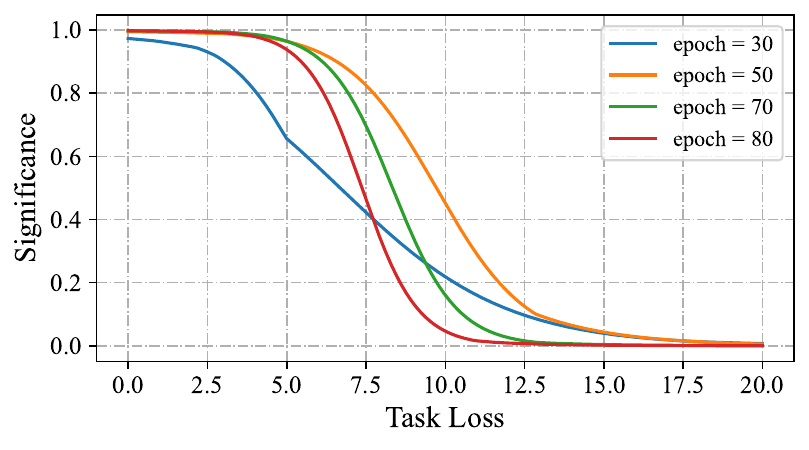}
	\caption{The significance evaluation function evolution under KB with flipping noise.}    
	\label{evol-nr}
	\vspace{-16pt}
\end{figure}

Fig. \ref{evol-nr} shows the evolution of  SEF in SCM under label flipping noise with FNR is $p=0.4$. 
Samples with large losses are suppressed, which indicates that the SCM adopts the weighting strategy to assign smaller significance to these high-loss samples, which are likely biased and can mislead the semantic coding network.  
In contrast, samples with smaller task losses are more likely to be pure and unbiased.

\begin{figure}[t]
    \subfloat[Loss distribution for TA-DeepJSCC.]{
        \includegraphics[width=0.9\columnwidth]{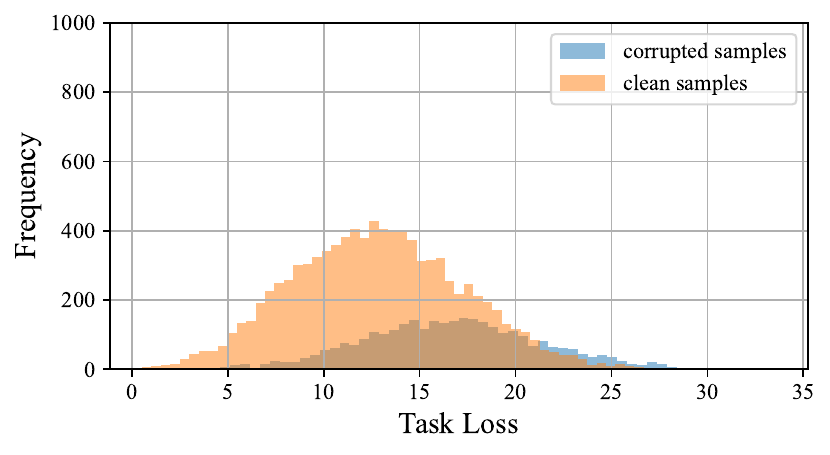}
        \label{fig:ta_deepjscc}
    } 
    \vspace{-8pt}
    \subfloat[Loss distribution for TALSC framework.]{
        \includegraphics[width=0.9\columnwidth]{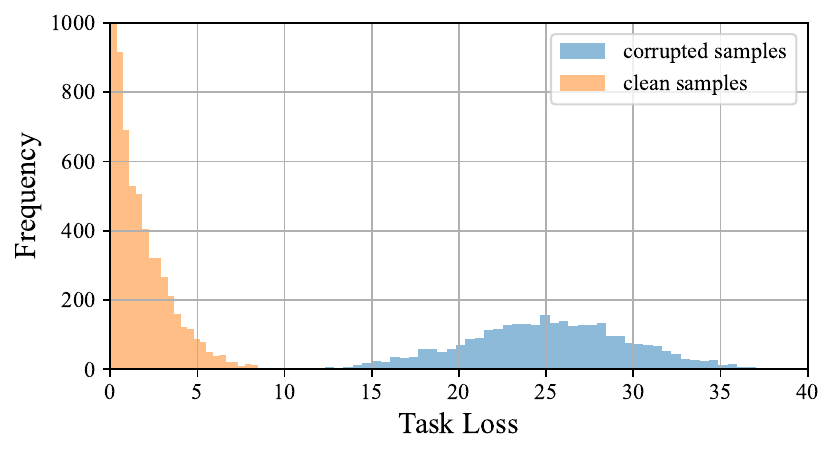}
        \label{fig:talsc}
    }
    \caption{Loss distribution of clean samples and flipped samples under FNR=0.2.}
    \label{loss-dist-nr}
    \vspace{-8pt}
\end{figure}

Fig. \ref{loss-dist-nr} illustrates the loss distribution of clean and flipped samples under label flipping noise. 
For the TA-DeepJSCC, the losses of clean and flipped samples are highly mixed and difficult to distinguish, as shown in Fig. \ref{fig:ta_deepjscc}. 
In contrast, with the TALSC framework, the significance evaluation effectively separates clean from flipped samples, suppressing high-loss flipped samples. 
As shown in Fig. \ref{fig:talsc},  the loss distributions become clearly separated, demonstrating the improved ability of the model to differentiate between clean and noisy data.

\subsubsection{SCM Parameterization: KAN vs. MLP}

\begin{figure}[t]
	\centering
	\includegraphics[width=0.98\columnwidth]{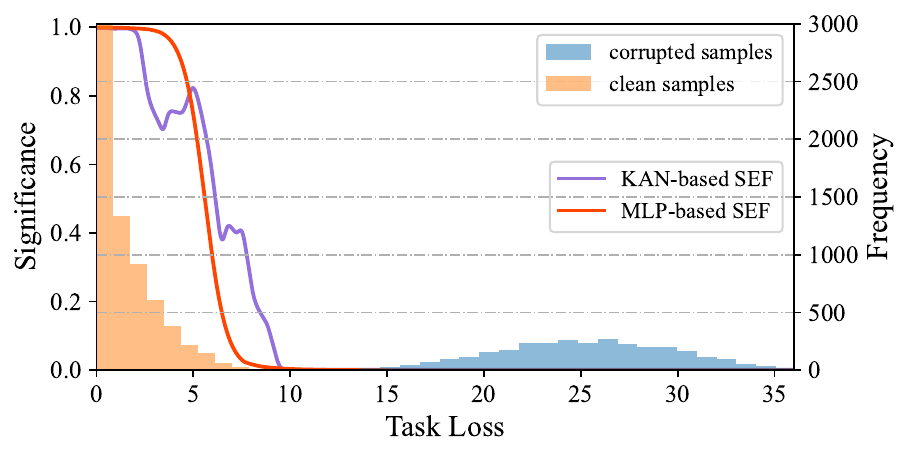}
	\caption{Comparison of the significance evaluation function obtained by parameterizing with MLP and KAN respectively under FNR=0.4.}    
	\label{kan_mlp}
	\vspace{-8pt}
\end{figure}

\begin{figure}[t]
	\centering
	\includegraphics[width=0.9\columnwidth]{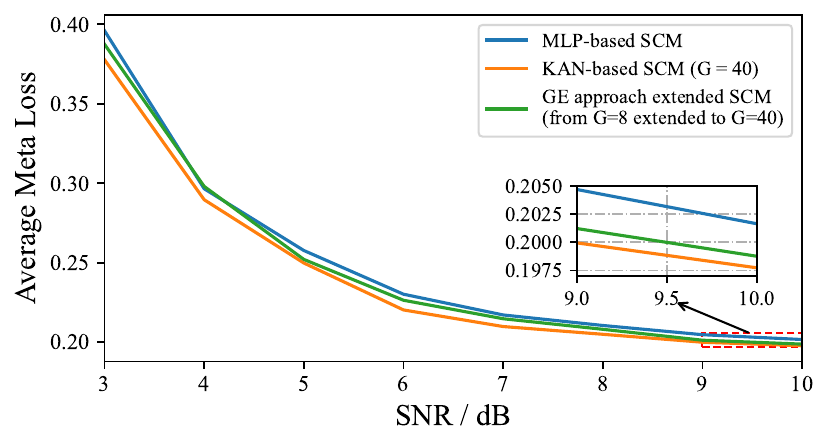}
	\caption{Meta-loss comparison of different parameterization schemes for SCM.}    
	\label{kan_ge_perform}
	\vspace{-8pt}
\end{figure}

Fig. \ref{kan_mlp} compares the SEF obtained by parameterizing the SCM with MLP and KAN under label flipping noise. 
It can be observed that the MLP-based SCM exhibits a relatively simple form, while the KAN-based SCM captures a more complex function that resembles a step function. 
This difference arises from the strong nonlinear representation capability of the KAN network.

Moreover, while the MLP-based SCM can achieve a similar level of function representation precision by increasing the number of hidden layer nodes, this requires retraining to obtain new parameter sets. 
In contrast, the KAN-based SCM incorporates GE approach, enabling it to adjust the precision of SCM without the need for retraining, thus offering a more efficient and flexible solution for \emph{meta-learner} optimization.
Fig. \ref{kan_ge_perform} further illustrates the meta-loss comparison of SEFs parameterized using schemes as follows:

a) MLP-based SCM: trained for 80 epochs.

b) KAN-based SCM: trained for 80 epochs with grid size $G = 40$.

c) GE approach extended SCM: first train with a KAN with grid size $G = 8$ for 75 epochs, then the GE approach extends the grid size to $G = 40$ by Eq. (\ref{ge}), and finally fine-tune for 5 epochs.

It is worth noting that a) MLP-based SCM and b) KAN-based SCM schemes possess a comparable number of parameters. The results show that the SCM parameterized by KAN has stronger representation ability and can better reduce the meta-loss, as shown in Fig. \ref{kan_ge_perform}. 
The c) GE approach extended SCM scheme shows that GE approach can effectively improve training efficiency while maintaining similar performance to b) KAN-based SCM scheme.

\section{Conclusion} \label{conclusion}

In this paper, we have proposed the TALSC framework for robust image semantic transmission, which integrated an SCM as the \emph{meta-learner} and a semantic coding networks as the \emph{learner}.    
Our TALSC framework has leveraged a well-designed SEF in SCM to map task loss to sample significance, enabling iterative updates of the \emph{learner} to suppress low-significance samples arising from label flipping noise and class imbalance, while emphasizing task-relevant samples.
To optimize the significance evaluation process in SCM, we have formulated a bi-level optimization problem and refined the update mechanism in meta-learning.    
For the construction of SEF in SCM, we have first used MLPs due to their practicality.     
Recognizing the efficiency and flexibility limitations of MLPs, we have proposed leveraging KANs to implement the SEF, and utilized their spline-based structure to reduce meta-loss in the optimization of the \emph{meta-learner}.     
Furthermore, to enhance flexibility and adaptability for SEF, we have proposed the GE approach for KAN-based SCM, which begins by constructing a coarse grid SEF using a low-dimensional parameter set.     
Then, we have derived an analytical transformation matrix to extend the coarse grid SEF to a finer resolution, guided by the MMSE criterion.     
Subsequent fine-tuning further refined the SEF to achieve optimal meta-loss performance.
Through extensive simulations, we have demonstrated that the TALSC framework effectively mitigates KB biases, achieving high SRA.     
Additionally, the framework can adapt to diverse precision requirements for the SEF, while the GE approach ensures computational efficiency and scalability for the SCM.     
These results have highlighted the potential of the TALSC framework in improving semantic transmission performance across a wide range of practical scenarios, offering robust solutions to the challenges posed by label noise, class imbalance, and task variability.

\bibliographystyle{IEEEtran}

\end{document}